\documentclass{llncs}

\usepackage{amssymb}
\usepackage{amsmath}
\usepackage{color}
\newcommand{\F}{{\mathbb F}}
\newcommand{\fp}{{\mathbb F}_{p}}
\newcommand{\fpn}{{\mathbb F}_{p^n}}
\newcommand{\fpk}{{\mathbb F}_{p^k}}
\newcommand{\fpm}{{\mathbb F}_{p^m}}
\newcommand{\ftwo}{{\mathbb F}_2}
\newcommand{\ftwon}{{\mathbb F}_{2^n}}
\newcommand{\fthree}{{\mathbb F}_{3}}
\newcommand{\fthreen}{{\mathbb F}_{3^n}}

\newcommand{\fthreek}{{\mathbb F}_{3^k}}
\newcommand{\Tr}{{\operatorname{Tr}}}

\newtheorem{construction}{Construction}

\begin{document}

\pagestyle{plain}

\title{Results on cubic bent and weakly regular bent $p$-ary functions leading to a class of cubic ternary not weakly regular bent functions}
\titlerunning{Ternary Binomial Bent Functions}
\author{Claude Carlet$^{\dag\ddag}$\thanks{Orcid 0000-0002-6118-7927} and Alexander Kholosha$^\ddag$}

\institute{$^\dag$University of Paris 8, Saint-Denis, France\\
$^\ddag$The Selmer Center, Department of Informatics, University of Bergen\\
\email{claude.carlet@gmail.com, oleksandr.kholosha@uib.no}} 

\maketitle
\thispagestyle{plain}

\begin{abstract}
Much work has been devoted to bent functions in odd characteristic, but there still remains a gap between our knowledge of binary and 
nonbinary bent functions. In the first part of this paper, we attempt to partially bridge this gap by generalizing to any characteristic 
important properties known in characteristic two concerning the Walsh transform of derivatives of bent functions. Some of these 
properties generalize to all bent functions, while others appear to apply only to weakly regular bent functions. We deduce a method to 
obtain a bent function by adding a quadratic function to a weakly regular bent function. We also identify a particular class of bent 
functions possessing the property that every first-order derivative in a nonzero direction has a derivative (which is then a second-order 
derivative of the function) equal to a nonzero constant. We show that this property implies bentness and is shared in particular by all 
cubic bent functions. This generalizes to the odd characteristic the notion of cubic-like bent function, that was introduced and studied 
for binary functions by Irene Villa and the first author. 

In the second part of the paper, we provide (for the first time) a primary construction leading to an infinite class of cubic ternary 
vectorial bent functions that have only not weakly regular components. We show the bentness of the component functions by two approaches: 
by calculating the Walsh transform directly and by considering the second-order derivatives (and applying the results from the first part 
of the paper). We prove that they are not weakly regular by showing they do not have one of the properties that we proved in the first 
part of the paper for weakly regular bent functions. 
\end{abstract}

\section{Introduction}
 \label{sec:intro}
Boolean bent functions were first introduced by Rothaus in 1976 as an interesting combinatorial object. They are functions from $\ftwo^n$ 
to $\ftwo$ for $n$ even, whose Hamming distance to the set of all affine functions is maximum. Equivalently, their squared Walsh transform 
takes constant value. They are characterized by the property that their derivatives in nonzero directions are balanced (functions having 
this property are called perfect nonlinear). Later, the research in this area was stimulated by the significant relation to the following 
topics in mathematics and computer science: algebraic combinatorics (with difference sets), coding theory (with Kerdock codes), sequences 
(with bent sequences) and cryptography (in the design of stream ciphers and $S$-boxes for block ciphers, some of the latter explicitly 
including bent functions). Kumar, Scholtz and Welch in \cite{KuScWe85} generalized the notion of Boolean bent functions to the case of 
functions from $(\mathbb Z/q\mathbb Z)^n$ to $\mathbb Z/q\mathbb Z$, where $q$ is the power of a prime $p$ (which covers the case of 
functions from $\fp^n$ to $\fp$, called $p$-ary functions, which are the subject of the present paper). Complete classification of bent 
functions looks hopeless even in the binary case and the case of $p$-ary bent functions is still more complicated. However, many explicit 
methods are known for constructing bent functions, either from scratch (with so-called primary constructions) or from other bent functions 
(in so-called secondary constructions). Diverse properties, as well as primary and secondary constructions, are known for binary bent 
functions. Many of these properties have been extended to $p$-ary functions, but some have never been investigated in this framework, maybe 
because they are a little more delicate to handle than in the binary case. In the first part of this paper, we generalize those dealing 
with the derivatives of bent functions (and in some cases, of their duals). Some properties extend to all bent functions (possibly of some 
particular algebraic degree) and some only extend to weakly regular bent functions. In a second part of the paper, we introduce a general 
class of ternary vectorial bent functions that have only not weakly regular components. We prove their bentness (and that they are not 
weakly regular) in two ways: by the very definition of bentness and by using the results from the first part of the paper. 

\section{Preliminaries}
 \label{sec:prel}
In this section, basic definitions and properties are recalled. Then we present the details on what is known as Maiorana-McFarland (MMF) 
functions. Finally, a new construction of $p$-ary bent functions is suggested (straightly generalized from the Boolean case), that also 
covers a particular class of MMF bent functions. 

\subsection{Definitions and Basic Properties}
 \label{ssec:def}
Take a prime $p$. Functions $f:\fp^n\mapsto\fp$ are called $p$-ary functions; if $p=2$ they are referred to as Boolean functions. Functions 
$F:\fp^n\mapsto\fp^k$ with $k\geq 2$ are called $p$-ary vectorial functions. For a fixed nonzero $\alpha\in\fp^k$, the $p$-ary function 
$f_{\alpha}(x)=\alpha\cdot F(x)$ is called a \emph{component function} of $F$ (here '$\cdot$' is an inner product). Every $p$-ary function 
has a unique multivariate representation called the Algebraic Normal Form (ANF) as a polynomial in $x_1,\ldots,x_n$ over $\fp$, where the 
variables $x_i$ occur with exponents smaller than $p$. The degree of the ANF of $f$ is called the algebraic degree of $f$. Functions of 
algebraic degree $2$ are called quadratic and those of algebraic degree $3$ are called cubic. The algebraic degree of a vectorial function 
is the largest algebraic degree of its component functions. 

Given $n$ and $p$, the vector space $\fp^n$ can be endowed with the structure of the field $\fpn$ (which allows to define polynomial 
functions). The finite field $\fpk$ is a subfield of $\fpn$ if and only if $k$ divides $n$. The trace mapping from $\fpn$ to the subfield 
$\fpk$ is defined by $\Tr_k^n(x)=\sum_{i=0}^{n/k-1}x^{p^{ik}}$. In the case when $k=1$, we use the notation $\Tr_n(x)$ instead of 
$\Tr_1^n(x)$. Every function $f:\fpn\mapsto\fp$ has a unique univariate representation and (not unique) trace representation as 
\[f(x)=\sum_{i=0}^{p^n-1} a_i x^i =\Tr_n(F(x))\quad\mbox{with}\quad a_i \in\fpn, a_{pi}=a_i^p\enspace.\]
Trace  representation becomes unique when transformed to 
\[f(x)=\sum_{j\in\Gamma_n}\Tr_{o(j)}(a_j x^j)+a_{p^n-1} x^{p^n-1}\enspace,\]
where $\Gamma_n$ is the set of coset leaders of the cyclotomic classes of $p$ modulo $p^n-1$, $o(j)$ is the size of the cyclotomic class 
containing $j$, $a_j\in\F_{p^{o(j)}}$ and $a_{p^n-1}\in\fp$. We call it the {\em relative trace form} of $f$. The algebraic degree of $f$ 
is equal to the maximal $p$-weight of the exponent $i$ with $a_i\neq 0$ in the univariate form of $f$ (or in its relative trace form), 
where the $p$-weight of $i=\sum_{j=0}^n i_j\, p^j$ (where $i_j\in \{0,\dots ,p-1\}$) equals $\sum_{j=0}^n i_j$. 

Given a function $f(x)$ mapping $\fpn$ to $\fp$, the {\em Walsh transforms} of $f$ and its inverse are respectively defined by the 
following identities
\begin{equation}
 \label{eq:W_t}
W_f(y)=\sum_{x\in\fpn}\omega^{f(x)-\Tr_n(xy)}; y\in\fpn\;\mbox{and}\; \omega^{f(x)}=\frac{1}{p^n}\sum_{y\in\fpn}W_f(y)\omega^{\Tr_n(xy)}
\end{equation}
where $\omega$ is the complex primitive $p^{\rm{th}}$ root of unity equal to $e^{\frac{2\pi i}{p}}$, and where the elements of $\fp$ are 
considered as integers modulo $p$. 

As defined in \cite{KuScWe85}, $f$ is a {\em $p$-ary bent function} if all its Walsh transform coefficients satisfy $|W_f(y)|^2=p^n$. 
According to \cite[Property~8]{KuScWe85}, it can be readily seen that for a $p$-ary bent function $f$ with odd $p$ holds 
\begin{equation}
 \label{eq:bent_Wtc}
p^{-n/2}W_f(y)=\left\{\begin{array}{ll}
\pm\,\omega^{f^*(y)},&\ \mbox{if}\ n\ \mbox{is even or}\ n\ \mbox{is odd and}\ p\equiv 1\pmod 4\\
\pm\,i\,\omega^{f^*(y)},&\ \mbox{if}\ n\ \mbox{is odd and}\ p\equiv 3\pmod 4\enspace,
\end{array}\right.
\end{equation}
where $i$ is a complex primitive fourth root of unity. We call $f^*:\fpn\mapsto \fp$ the {\it dual} of $f$. Vectorial function is called 
bent if all its component functions are bent.
 
Function $f$ is bent if and only if all the derivatives $D_af(x)=f(x+a)-f(x)$ in nonzero directions $a$ are balanced, that is, take any 
value in $\fp$ equally often (i.e., $f$ is {\em perfect nonlinear}). This was observed already in \cite[Eq.~(15)]{KuScWe85} and follows 
from 
\begin{align}
 \label{eq:relwd}
\nonumber |W_f(y)|^2&=\sum_{u,v\in\fpn}\omega^{f(u)-f(v)+\Tr_n(y(v-u))}\\
\nonumber&=\sum_{w\in\fpn}\omega^{\Tr_n(yw)}\sum_{v\in\fpn}\omega^{f(v-w)-f(v)}\\
&=\sum_{w\in\fpn}\omega^{\Tr_n(yw)}\sum_{v\in\fpn}\omega^{D_{-w}f(v)}\enspace.
\end{align}
If $f$ is perfect nonlinear meaning $D_{-w}f$ is balanced for $w\neq 0$ (according to \cite[Theorem~5]{CaDi04}), then the inner sum is zero  
(and conversely, according to Lemma~\ref{le:omega} below) and $f$ is bent. The converse follows from \cite[Theorem~16]{CaDi04}. Similarly, 
vectorial function is perfect nonlinear if its derivatives in any nonzero direction are balanced. 

A $p$-ary bent function $f(x)$ is called {\em regular} (see \cite[Definition~3]{KuScWe85} and \cite[p.~576]{Ho04_1}) if, for every 
$y\in\fpn$, the normalized Walsh coefficient $p^{-n/2}W_f(y)$ is equal to a complex $p^{\rm{th}}$ root of unity, i.e., 
$p^{-n/2}W_f(y)=\omega^{f^*(y)}$, that is, in (\ref{eq:bent_Wtc}), we necessarily have $n$ even or $n$ odd with $p\equiv 1\pmod 4$ and the 
sign $\pm$ is in fact $+$. A bent function $f(x)$ is called {\em weakly regular} if there exists a complex $u$ having unit magnitude such 
that $u^{-1}p^{-n/2}W_f(y)=\omega^{f^*(y)}$ for all $y\in\fpn$, (and by (\ref{eq:bent_Wtc}) we can take $u=\pm 1$ or $u=\pm i$). The dual 
of a (weakly) regular bent function is again a (weakly) regular bent function, but the dual of not weakly regular bent functions is not 
necessarily bent and we call a bent function {\em dual-bent} if its dual is bent, so all weakly regular bent functions are dual bent. By 
the inverse Walsh transform, $f^{**}(x)=f(-x)$ for weakly regular bent functions and (as shown in \cite{OzPe20}) also holds for not weakly 
regular dual-bent functions. Therefore, the dual of a not weakly regular dual-bent function is again not weakly regular bent. Also, 
$f^{***}(x)=f^*(-x)$ and $f^{****}(x)=f(x)$. Both dual-bent and not dual-bent not weakly regular bent functions exist, see 
\cite{CeMeiPo13_1,CeMeiPo16}. 

\begin{definition}
 \label{de:par_bent}
The linear kernel of a function $f:\fpn\mapsto\fp$ is the set of those $a\in\fpm$ such that $D_af$ is a constant function (linear kernel is 
a subspace of $\fpn$). Function $f$ is called partially bent if for any $a\in\fpm$, the derivative $D_a f(x)$ is either balanced or 
constant. 
\end{definition}

An important notion for the study of $p$-ary functions is that of equivalence. Each time a bent function is found, all the functions that 
are equivalent to it are also bent (and some of them may have complex representations while the function found has in general a simple 
representation). When a new class is found, its discoverers must check whether some of its elements are inequivalent to all known bent 
functions found until then. This can be very difficult and even almost impossible when too many classes are known (which is the case in 
binary). 

\begin{definition}
Functions $f,g:\fp^n\mapsto\fp$ are \emph{extended-affine equivalent} (in brief, EA-equivalent) if there exists an affine permutation $L$ 
of $\fp^n$, an affine function $l:\fp^n\mapsto\fp$ and an element $a$ of $\fp^*$ such that $g(x)=a(f\circ L)(x)+l(x)$. A class of functions 
is \emph{complete} if it is a union of EA-equivalence classes. The \emph{completed class} is the smallest possible complete class that 
contains the original one. 
\end{definition}
The families of weakly regular bent functions, of not weakly regular bent functions, and of dual-bent functions are invariant under 
EA-equivalence (see \cite{CeMeiPo13_1}). 

It is interesting to find bent functions allowing few terms in their relative trace form. Quadratic bent functions are well understood 
since every such function is (weakly) regular and is EA-equivalent to either $x_1^2+\cdots+x_n^2$ or $dx_1^2+\cdots+x_n^2$ where $d$ is a 
nonsquare in $\fp$ (see \cite[Theorem~3]{Mei22}). From non-quadratic, the most common are regular bent functions with the trace form 
consisting of (so called, Dillon type) exponents divisible by $p^{n/2}-1$. Such bent functions are likely to exist with any number of terms 
(see, e.g, \cite{QiTaHu18,TaQiHu19,TaXuQiZh21}) and are included in the class of the partial spread functions (see \cite{LiLU14}). In 
addition to quadratic and ternary Dillon bent monomials, the only known are ternary Coulter-Matthews perfect nonlinear vectorial functions 
$x^{(3^k+1)/2}$ on $\fthreen$ with $k$ odd and $\gcd(k,n)=1$ (see \cite{CoMa97}) and also ternary functions with $n=2k$ and $k$ odd given 
by $f(x)=\Tr_n\Big(ax^{\frac{3^n-1}{4}+3^k+1}\Big)$, where $a=\xi^{\frac{3^k+1}{4}}$ and $\xi$ is a primitive element of $\fthreen$ (see 
\cite{HeKh06_1}), both classes give weakly regular bent functions. Moreover, it is conjectured that the currently known list of $p$-ary 
bent monomials, including the sporadic case in Table~\ref{ta:1}, is complete.\footnote{As in the binary case, this kind of conjecture may 
be invalidated later, as many conjectures on Boolean and vectorial functions were in the past.} 

In addition to quadratic and Dillon cases, the only known class of binomials is for $n=4k$ with 
$f(x)=\Tr_n\big(x^{p^{3k}+p^{2k}-p^k+1}+x^2\big)$ proven in \cite{HeKh10_4} to be a weakly regular bent function for any odd $p$. Also, for 
some Coulter-Matthews bent functions, the dual is trinomial (see \cite{HuZhSh17}). In Section~\ref{sec:vec}, we present an infinite class 
of trinomial ternary cubic vectorial bent functions and that have only not weakly regular components and is the first of this kind known. 
We prove bentness of component functions by two approaches. The first approach is by the definition: we calculate the Walsh transform. The 
second approach is by the second-order derivatives. 

In Section~\ref{sec:cubic}, we prove that if a function (cubic or not) is such that every of its first-order derivatives $D_a f$ with 
$a\neq 0$, admits a derivative $D_a D_bf(x)=f(x+a+b)-f(x+a)-f(x+b)+f(x)$ (we shall also denote it by $D_{a,b}f$) that is constant and 
nonzero then it is bent. In the case the function is cubic (i.e., has algebraic degree $3$) this sufficient condition is also necessary and 
since our functions are cubic, we know that if the functions are bent then this proof by the derivatives is possible. We also generalize in 
Section~\ref{sec:deriv} to every characteristic $p$ some properties of the first-order derivatives that were known only for binary 
functions. These properties are valid for weakly regular bent functions only, and they represent then tools for showing that a given bent 
function is not weakly regular as we do in Theorem~\ref{th:trinom_nwr_vec}. 

\subsection{Constructions}
 \label{ssec:constr}
Infinite classes of not weakly regular bent functions were constructed using \mbox{(semi-)}direct sum \cite{TaYaZh10,CeMeiPo16}, 
generalized Rothaus construction \cite{Mei16}, and MMF construction \cite{CeMcGMei12,CeMeiPo13_2}. The MMF original class provides a 
primary construction of regular bent functions. All the known constructions of not weakly regular bent functions are secondary since they 
build on some functions known to be (partially) bent. Primary constructions are not known. 

A well-known primary construction of Boolean functions is that of Maiorana-McFarland (see, e.g., \cite[Definition~46]{Ca96}). The original 
version of this construction (from \cite{Di74}) was for designing bent functions only and represented the input as a pair of vectors of the 
same length. Later was identified a general structure of some Boolean functions where the input is a pair of vectors with possibly 
different lengths. The truth table of these functions is the concatenation of the truth tables of affine functions. This extension allowed 
to construct bent functions but also other types of cryptographically interesting functions. We recall now the $p$-ary generalization of 
this class and we specify what is the generalization of the original class. 

\begin{definition}
 \label{de:MMF}
Let $n$ and $k$ be any positive integers such that $k<n$. We call MMF function any $n$-variable $p$-ary function of the form
\begin{equation}
 \label{eq:MMFform}
f(x,y)=x\cdot\psi(y)+h(y)\enspace, 
\end{equation}
where $\psi:\fp^{n-k}\mapsto\fp^k$, $h:\fp^{n-k}\mapsto\fp$, and '$\cdot$' is an inner product. The subclass with $n$ even and $k=n/2$ is 
called the MMF original class.  
\end{definition}

Any function $f:\fp^n\mapsto\fp$ that is affine on the cosets of a $k$-dimensional subspace $U$ is affine equivalent to a function of the 
form (\ref{eq:MMFform}). Affine means that for any $c\in\fp^n$ and $z\in U$ we have $f(c+z)=z\cdot A_c+B_c$ for some constants 
$A_c\in\fp^n$ and $B_c\in\fp$ that are defined by $c$. A standard criterion for a Boolean function to be in the completed MMF class can be 
formulated in terms of the second-order derivatives (see \cite[Proposition~54]{Ca96}). This criterion can be naturally extended to the 
$p$-ary case. Namely, $f:\fp^n\mapsto\fp$ belongs to the completed MMF class if and only if there exists a $k$-dimensional subspace $U$ of 
$\fp^n$ such that the second-order derivatives of $f$ in the direction of elements of $U$ vanish, i.e., $D_{a,b}f(x)\equiv 0$ for any 
$a,b\in U$. We say then that $f$ belongs to the {\em completed class of MMF functions with associated ($k$-dimensional) subspace $U$}. 

It is well known that an MMF function in the original class is (regular) bent if and only if $\psi$ is a permutation. On the opposite, when 
$k\neq n/2$, infinite classes of $\psi()$ and $h()$ giving bent function $f$ are not known. MMF bent functions can be (weakly) regular and 
not weakly regular in even and odd dimension (see Construction~\ref{co:1}). The only known construction giving MMF bent functions beyond 
the original class, is using partially-bent functions as follows (see \cite[Theorem~2]{CeMeiPo13_2}). 

\begin{theorem}[concatenation of partially-bent functions]
 \label{th:MMFconstr}
Take function $f(x,y):\fp^k\times\fp^{n-k}\mapsto\fp$, with $k\leq n/2$ and where $f_x(y)=f(x,y)$ are partially-bent functions with the 
common $k$-dimensional linear kernel $U$ and 
\[\mathrm{supp}(W_{f_{x_1}})\bigcap\mathrm{supp}(W_{f_{x_2}})=\emptyset\]
whenever $x_1\neq x_2$. Then $f$ is a bent function in the completed MMF class and has associated subspace $U$ (seen as a subspace in 
$\fp^n$ identified with $\{0\}\times U$). 
\end{theorem}

Vectorial Construction~\ref{co:1} comes from \cite[Theorem~3]{CeMeiPo20}. With $k=1$, it can be seen as a particular case covered both by 
Theorem~\ref{th:MMFconstr} (concatenation of partially-bent functions) and by (\ref{eq:co2_sp}) in Construction~\ref{co:2} (concatenation 
of bent functions). Construction~\ref{co:1} is interesting since it has no additional restrictions on the (partially-)bent functions used 
in concatenation. 

\begin{construction}
 \label{co:1}
Take positive integers $n,k$ and let $\{F_c\ |\ c\in\fpk\}$ be a set of vectorial bent functions mapping $\fpk^n$ to $\fpk$. The Walsh 
transform coefficients of a ($p$-ary bent) component function $f_{a,c}=\Tr_k(aF_c)$ for every $a\in\fpk^*$ are 
\[W_{f_{a,c}}(y)=\sum_{x\in\fpk^n}\omega^{\Tr_k(a F_c(x)-y\cdot x)}=p^{kn/2}u_{a,c}(y)\omega^{f^*_{a,c}(y)}\enspace,\] 
where $y\in\fpk^n$, $|u_{a,c}(y)|=1$ and '$\cdot$' is an inner product of $\fpk^n$. Then vectorial function $F:\fpk^{n+2}\mapsto\fpk$ 
defined as 
\[F(x,x_{n+1},x_{n+2})=F_{x_{n+2}}(x)+x_{n+1}x_{n+2}\]
is bent. The Walsh transform of a component function $f_a=\Tr_k(aF)$ for any $a\in\fpk^*$ is equal to 
\[W_{f_a}(y,y_{n+1},y_{n+2})=p^{(n+2)k/2}u_{a,a^{-1}y_{n+1}}(y)\omega^{f^*_{a,a^{-1}y_{n+1}}(y)-\Tr_k(a^{-1}y_{n+1}y_{n+2})}\enspace,\]
where $y\in\fpk^n$. If $f_{a,c}$ all are weakly regular (then $u_{a,c}(y)=u_{a,c}$ do not depend on $y$) but not all $u_{a,c}$ are the 
same, then $f_a$ is not weakly regular. 
\end{construction}

The following construction is the $p$-ary generalization of \cite[Theorem~15]{Ca20} that appeared originally in \cite{Ca96}. As far as we 
know, this has not been suggested before. Example of an MMF bent function obtained by concatenation of bent functions that is not covered 
by Theorem~\ref{th:MMFconstr}, is in Proposition~\ref{pr:bc_ex}. 

\begin{construction}[concatenation of bent functions]
 \label{co:2}
Take positive integers $n$ and $m$. Let $f$ be a $p$-ary function on $\F^{n+m}_p=\F^n_p\times\F^m_p$ such that, for every $y\in\F^m_p$, the 
function $f_y(x):\F^n_p\mapsto f(x,y)$ is bent with $W_{f_y}(s)=u(s)p^{n/2}\omega^{f^*_y(s)}$ and $u(s)\in\{\pm 1,\pm i\}$. Then $f$ is 
bent if and only if, for every $s\in\F^n_p$, the function $\phi_s(y)=f^*_y(s)$ is bent on $\F^m_p$. If this condition is satisfied, then 
the dual of $f$ is the function $f^*(s,t)=\phi^*_s(t)$ (taking as inner product in $\F^n_p\times\F^m_p : (x,y)\cdot(s,t)=x\cdot s+y\cdot 
t$). 

In particular, if $m=2d$ and 
\begin{equation}
 \label{eq:co2_sp}
f_{y_1,y_2}(x)=g_{y_2}(x)+y_1\cdot\pi(y_2)
\end{equation}
with $y_1,y_2\in\fp^d$, $g_c:\fp^n\mapsto\fp$ for every $c\in\fp^d$, and $\pi()$ is a permutation of $\fp^d$ then the condition on 
$f^*_y(s)$ as a function of $y$ to be bent is automatically satisfied, and the values of $u(s)$ are not relevant (here $y=(y_1,y_2)$). 
Obviously, $f(x,y_1,y_2)=f_{y_1,y_2}(x)$ is an MMF function. Function $f$ is (dual) bent if and only if $g_c$ are (dual) bent for all 
$c\in\fp^d$. Case~(\ref{eq:co2_sp}) is also covered by Theorem~\ref{th:MMFconstr}.
\end{construction}

\begin{proof}
This construction holds in general since
\begin{align*}
&W_f(s,t)=\sum_{y\in\F^m_p}W_{f_y}(s)\omega^{-t\cdot y}=u(s)p^{n/2}\sum_{y\in\F^m_p}\omega^{\phi_s(y)-t\cdot y}\\
&=u(s)p^{n/2}W_{\phi_s}(t)=u(s)v_s(t)p^{(n+m)/2}\omega^{\phi^*_s(t)}\enspace, 
\end{align*}
where $W_{\phi_s}(t)=v_s(t)p^{m/2}\omega^{\phi^*_s(t)}$. Construction~\ref{co:2} can provide (weakly) regular and not weakly regular bent 
functions $f$ depending on $u(s)v_s(t)$ that are characteristics of $f_y(x)$ and $\phi_s(y)=f^*_y(s)$. For instance, $f$ is weakly regular 
if $u(s)v_s(t)$ is constant when $(s,t)\in\F^n_p\times\F^m_p$. This will hold, in particular, if both $f_y$ with $y\in\F^m_p$ and $\phi_s$ 
with $s\in\F^n_p$ are all regular bent that gives $u(s)=v_s(t)\equiv 1$. Construction~\ref{co:2} in particular form (\ref{eq:co2_sp}) will 
cover a subset of not weakly regular bent functions that we find in Theorem~\ref{th:trinom_nwr_vec} below. 

In the special case when $m=2d$ and functions $f_y$ have the form of (\ref{eq:co2_sp}), $f_{y_1,y_2}$ is bent if and only if $g_{y_2}$ is 
bent and 
\[f^*_{y_1,y_2}(s)=g_{y_2}^*(s)+y_1\cdot\pi(y_2)=\phi_s(y_1,y_2)\enspace.\]
Then, $\phi_s$ is bent in the original MMF class which dual is well known. For any $c\in\fp^d$, let 
$W_{g_c}(s)=u_c(s)p^{n/2}\omega^{g^*_c(s)}$ for $s\in\F_p^n$ to obtain 
\begin{align*}
W_f(s,t_1,t_2)&=\sum_{y_1,y_2\in\F^d_p}W_{f_{y_1,y_2}}(s)\omega^{-t_1\cdot y_1-t_2\cdot y_2}\\
&=p^{n/2}\sum_{y_1,y_2\in\F^d_p}u_{y_2}(s)\omega^{g_{y_2}^*(s)+(\pi(y_2)-t_1)\cdot y_1-t_2\cdot y_2}\\
&=p^{(n+2d)/2}u_{\pi^{-1}(t_1)}(s)\omega^{g^*_{\pi^{-1}(t_1)}(s)-\pi^{-1}(t_1)\cdot t_2}\enspace.
\end{align*}
If $g^*_c$ are bent for every $c\in\fp^d$ (i.e., $g_c$ are dual bent) then $f^*(s,t_1,t_2)=g^*_{\pi^{-1}(t_1)}(s)-\pi^{-1}(t_1)\cdot t_2$ 
also has the form of (\ref{eq:co2_sp}) (substitute $t_2$ with $-t_2$ which is EA-equivalent transformation) and is also bent, thus, $f$ is 
dual bent. Note that $|W_{f^*}(s,t_1,t_2)|^2=p^{2d}|W_{g^*_{\pi(-t_1)}}(s)|^2$ so $f^*$ is bent if and only if $g^*_{\pi(-t_1)}$ is bent 
for any $t_1\in\fp^d$. Thus, $f$ is dual bent if and only if all $g_c$ are dual bent. 

Also, if $g_{y_2}$ are bent for any $y_2\in\fp^d$ then functions $f_{y_2}(x,y_1)=f_{y_1,y_2}(x)$ on $\fp^{n+d}$ defined as in 
(\ref{eq:co2_sp}) are partially bent with $d$-dimensional linear kernel $\{0\}\times\fp^d$ and pairwise disjoint support of Walsh spectra 
(since the linear terms in variable $y_1$ are distinct). We conclude that special form (\ref{eq:co2_sp}) of Construction~\ref{co:2} that 
gives MMF bent functions is covered by Theorem~\ref{th:MMFconstr}. If $n=0$ then (\ref{eq:co2_sp}) becomes 
\[f(y_1,y_2)=g_{y_2}+y_1\cdot\pi(y_2)=g(y_2)+y_1\cdot\pi(y_2)\]
that is the original MMF bent function. Thus, any original MMF bent function can be obtained by a special form (\ref{eq:co2_sp}) of 
Construction~\ref{co:2}.\qed
\end{proof} 

Obviously, Construction~\ref{co:1} with $k=1$ is a particular case of (\ref{eq:co2_sp}). It takes a set of bent functions $g_c$ for every 
$c\in\fp$ to produce the following bent function $f$ of $x\in\F_p^n$ and $x_{n+1},x_{n+2}\in\fp$ 
\[f(x,x_{n+1},x_{n+2})=g_{x_{n+2}}(x)+x_{n+1}x_{n+2}\]
and this has the form of (\ref{eq:co2_sp}) with $d=1$ and $\pi(y)=y$.

Weakly regular bent functions were shown to be useful for constructing certain combinatorial objects such as partial difference sets, 
strongly regular graphs and association schemes (see \cite{TaPoFe10,PoTaFeLi11,HyLe19}). This justifies why the classes of weakly regular 
bent functions are of interest. But the question whether infinite classes of not weakly regular bent function with simple representation 
exist arises. Indeed, no primary construction for not weakly regular bent functions is known and we believe that this question must be 
addressed. Moreover, as we shall see with the generalizations we give of results which were previously known only for binary bent 
functions, weakly regular functions have hard constraints about their derivatives, and considering not weakly regular bent functions may 
allow more flexibility in future applications. For a comprehensive reference on $p$-ary bent functions (in particular, monomial and 
quadratic ones) we refer the reader to \cite{HeKh06_1} and to the general survey on $p$-ary bent functions \cite{Mei22}. 

Let $f(x)=\Tr_n(a_1 x^{d_1}+a_2 x^{d_2})$ and $\xi$ be a primitive element of $\fthreen$. Table~\ref{ta:1} contains sporadic (not 
classified) examples of ternary not weakly regular bent functions over $\fthreen$ having one or two terms in the trace form, found 
computationally. Here 'db' stands for dual-bent and 'ndb' for not dual-bent. The last column indicates whether the function belongs to the 
(completed) MMF class (see \cite[Theorems~4,5]{CeMeiPo13_2}). 
\begin{table}[h]
\caption{Sporadic examples of ternary not weakly regular bent functions}
 \label{ta:1}
\renewcommand{\arraystretch}{1.2}
\centering
\begin{tabular}{c|c|c|c|c|l|l}
\hline $\ n\ $&$\ a_1\ $&$\ d_1\ $&$\ a_2\ $&$\ d_2\ $&\\
\hline $3$&$1$&$8$&$1$&$14$&db&MMF\\
$4$&$1$&$22$&$\xi^{20}$&$4$&ndb&\\
$6$&$\xi^7$&$98$&&&ndb&not cMMF\\
$6$&$\xi^7$&$14$&$\xi^{35}$&$70$&ndb&\\
$6$&$\xi$&$20$&$\xi^{41}$&$92$&db&\\
\hline
\end{tabular}
\end{table}

The following result from \cite[Corollary~3]{HeKh06_1} is frequently used in the paper.

\begin{proposition}
Take an odd prime $p$. The $p$-ary function $f(x)=\Tr_n\left(a x^2\right)$ is a (weakly) regular bent function for any nonzero $a\in\fpn$ 
and the Walsh transform coefficient of $f$ at point $b\in\fpn$ satisfy 
\begin{equation}
 \label{eq:quad_Wtc}
\eta(a)(-1)^{n-1}p^{-n/2}W_f(b)=\left\{\begin{array}{ll}
\omega^{-\Tr_n\left(\frac{b^2}{4a}\right)},&\mbox{if}\ p\equiv 1\pmod 4\\
i^n\omega^{-\Tr_n\left(\frac{b^2}{4a}\right)},&\mbox{if}\ p\equiv 3\pmod 4\enspace,
\end{array}\right.
\end{equation}
where $i$ is the complex primitive fourth root of unity and $\eta$ is the quadratic character of $\fpn$. 
\end{proposition}

\begin{proposition}
 \label{pr:bc_ex}
Take a ternary function $f:\F_{3^3}\mapsto\fthree$ defined by 
\[f(x)=\Tr_3(x^8+x^{14}+\xi x^2)\enspace,\]
where $\xi=0$ or $\xi\in\F_{3^3}^*$ is a root of the primitive polynomial $x^3-x+1$ over $\fthree$. Then $f$ is an MMF not weakly regular 
dual-bent function. In the case when $\xi=0$, $f$ is obtained by concatenation of partially-bent functions and in the case when $\xi\neq 
0$, $f$ is a concatenation of quadratic bent functions (and not a concatenation of partially-bent functions but EA-equivalent to such 
concatenation). All three functions $f$ with $\xi\neq 0$ are EA-equivalent to the dual of $f$ with $\xi=0$. 
\end{proposition}

\begin{proof}
For $\xi\neq 0$, view the field $\F_{3^3}=\fthree(\xi)$ as a $3$-dimensional vector space over $\fthree$ with the basis $\{1,\xi,\xi^2\}$. 
Use computer algebra package to obtain the following ANF of $f$  
\begin{align*}
f(x_1,x_2,x_3)&=x_2^2 x_3^2+x_2^2-x_3^2+x_1 x_3&&\mbox{if}\quad\xi=0\\
f(x_1,x_2,x_3)&=x_2^2 x_3^2+x_2^2+x_2 x_3+x_1(x_2+x_3)&&\mbox{if}\quad\xi\neq 0\enspace.
\end{align*}
Taking $\xi=0$ gives the known sporadic example of binomial not weakly regular bent function that belongs to the MMF class (see 
Table~\ref{ta:1}) and obtained by concatenation of partially-bent functions having $1$-dimensional linear kernel (as in 
Theorem~\ref{th:MMFconstr}). 

Now assume $\xi\neq 0$. Obviously, $f$ is affine on every coset of the $1$-dimensional subspace $V=\{(x_1,0,0)\}$ of $\fthree^3$ so it is 
in the MMF class. Also, $f$ restricted to the $2$-dimensional subspace $W=\{(x_1,x_2,-x_2)\}$ of $\fthree^3$ containing $V$ and its cosets 
$D$ (identified by $x_2+x_3$), gives three partially-bent functions $g_{x_2+x_3}(x_1,x_2)$ with pairwise disjoint Walsh support. 
Specifically, $g_0(x_1,x_2)=x_2^2$, $g_1(x_1,x_2)=-x_2^2-x_2+x_1$, and $g_{-1}(x_1,x_2)=-x_2^2+x_2-x_1$. Thus, $f$ is EA-equivalent to the 
concatenation construction in Theorem~\ref{th:MMFconstr} and satisfies \cite[Theorem~2]{CeMeiPo13_2}. 

The subfunctions of $f$ obtained by fixing variable $x_3$ are 
\begin{align*}
f_0(x_1,x_2)&=f(x_1,x_2,0)=x_2^2+x_1 x_2\\
f_1(x_1,x_2)&=f(x_1,x_2,1)=-x_2^2+x_1 x_2+x_1+x_2\\
f_{-1}(x_1,x_2)&=f(x_1,x_2,-1)=-x_2^2+x_1 x_2-x_1-x_2
\end{align*}
and they are all quadratic non-degenerate (thus, bent) functions that are regular with the 'plus' sign. Calculate their Walsh transform 
using (\ref{eq:quad_Wtc})
\begin{align*}
W_{f_0}(s_1,s_2)&=\sum_{x_1,x_2}\omega^{x_2^2+x_1 x_2-s_1 x_1-s_2 x_2}=i 3^{1/2}\sum_{x_1}\omega^{-s_1 x_1-(x_1-s_2)^2}\\
&=i 3^{1/2}\sum_{x_1}\omega^{-x_1^2-(s_1+s_2)x_1-s_2^2}=3\omega^{(s_1+s_2)^2-s_2^2}=3\omega^{s_1^2-s_1 s_2}\enspace,\\
W_{f_1}(s_1,s_2)&=\sum_{x_1,x_2}\omega^{-x_2^2+x_1 x_2+x_1+x_2-s_1 x_1-s_2 x_2}=-i 3^{1/2}\sum_{x_1}\omega^{x_1-s_1 x_1+(x_1-s_2+1)^2}\\
&=-i 3^{1/2}\sum_{x_1}\omega^{x_1^2-(s_1-s_2)x_1+s_2^2+s_2+1}=3\omega^{s_2^2+s_2+1-(s_1-s_2)^2}\\
&=3\omega^{-s_1^2-s_1 s_2+s_2+1}\enspace,\\
W_{f_{-1}}(s_1,s_2)&=\sum_{x_1,x_2}\omega^{-x_2^2+x_1 x_2-x_1-x_2-s_1 x_1-s_2 x_2}=-i 3^{1/2}\sum_{x_1}\omega^{-x_1-s_1 x_1+(x_1-s_2-1)^2}\\
&=-i 3^{1/2}\sum_{x_1}\omega^{x_1^2-(s_1-s_2)x_1+s_2^2-s_2+1}=3\omega^{s_2^2-s_2+1-(s_1-s_2)^2}\\
&=3\omega^{-s_1^2-s_1 s_2-s_2+1}
\end{align*}
and
\begin{align*}
f_{y}^*(s_1,s_2)&=(1-y^2)(s_1^2-s_1 s_2)+y^2(-s_1^2-s_1 s_2+s_2 y+1))\\
&=(s_1^2+1)y^2+s_2 y+s_1^2-s_1 s_2=\phi_{s_1,s_2}(y)
\end{align*}
is bent as a function of $y$ since $s_1^2+1\neq 0$. Then
\begin{align*} 
W_f(s_1,s_2,s_3)&=3W_{\phi_{s_1,s_2}}(s_3)=\eta(s_1^2+1)i 3^{3/2}\omega^{s_1^2-s_1 s_2-(s_2-s_3)^2(s_1^2+1)^{-1}}\\
&=\eta(s_1^2+1)i 3^{3/2}\omega^{s_1^2-s_1 s_2-(s_2-s_3)^2(s_1^2+1)}
\end{align*}
proving that $f$ is not weakly regular bent with the sign equal to $\eta(s_1^2+1)$. This bent function is obtained by concatenation of bent 
functions. 

Substitute $s_3=s_2-s_3$ and multiply by $-1$ to obtain the following function EA-equivalent to the dual of $f$
\[f^*(s_1,s_2,s_3)=s_1^2 s_3^2-s_1^2+s_3^2+s_1 s_2\enspace.\]
This is a bent function EA-equivalent to $f$ with $\xi=0$ which proves that $f$ is dual-bent.\qed 
\end{proof}

\section{$p$-ary Cubic-Like Bent Functions}
 \label{sec:cubic}
Boolean cubic-like bent functions were introduced in \cite{CaVi25}. These are binary functions with the property that, for any 
$a\in\ftwon^*$, there exists $b\in\ftwon$ such that the second-order derivative $D_{a,b}f(x)=f(x+a+b)+f(x+a)+f(x+b)+f(x)$ equals constant 
function $1$. The 
name of cubic-like bent function is coherent with the facts that:\\
(1) this property is always satisfied by cubic bent functions, because the derivatives are then quadratic and thanks to 
\cite[Proposition~55]{Ca99},\\
(2) any Boolean function satisfying the property is bent since every derivative in a nonzero direction is then automatically balanced. In 
the following, we generalize this to the $p$-ary case. 

We will use the following known lemma frequently in what follows, and we provide a proof, for our paper to be self-contained. 

\begin{lemma}[\cite{HeKh10_4}]
 \label{le:omega}
A $p$-ary function $f$ mapping $\fpn$ to $\fp$ is balanced if and only if $W_f(0)=0$.
\end{lemma} 

\begin{proof}
It is well known that the polynomial $p(x)=\sum_{i=0}^{p-1}x^i$ is irreducible over the rational number field and $p(\omega)=0$. Thus, 
$p(x)$ is the minimal polynomial of $\omega$ over the rational numbers. Then, obviously, if $f$ is balanced then 
$W_f(0)=\sum_{i=0}^{p-1}N\omega^i=Np(\omega)=0$, where $N=N_i=\#\{x\in\fpn : f(x)=i\}$. Conversely, if 
$W_f(0)=\sum_{i=0}^{p-1}N_i\omega^i=0$ and $N_i$ are integers then $p(x)$ divides $\sum_{i=0}^{p-1}N_ix^i$ and thus, all $N_i$ are the same 
and $f$ is balanced.\qed 
\end{proof} 

\begin{definition}
A $p$-ary function $f$ mapping $\fpn$ to $\fp$ is called {\em cubic-like bent} if, for any $a\in\fpn^*$, there exists $b\in\fpn$ such that 
the second-order derivative $D_{a,b}f(x)$ is a nonzero constant function. 
\end{definition}

\begin{theorem}
 \label{th:3like_bent}
If a $p$-ary function $f$ mapping $\fpn$ to $\fp$ is cubic-like bent, then it is bent (i.e. perfect nonlinear). 
\end{theorem}

\begin{proof} 
Note that $D_af(x+b)=D_af(x)+D_{a,b}f(x)$ for any $a,b\in\fpn$. Take any $a\in\fpn^*$ and corresponding $b\in\fpn$ such that $D_{a,b}f(x)$ 
is a nonzero constant. Then 
\[\sum_{x\in\fpn}\omega^{D_af(x)}=\sum_{x\in\fpn}\omega^{D_af(x+b)}=\omega^{D_{a,b}f(0)}\sum_{x\in\fpn}\omega^{D_af(x)}\]
that is possible only if $\sum_{x\in\fpn}\omega^{D_af(x)}=0$ meaning that $D_af(x)$ is balanced according to Lemma~\ref{le:omega}. To 
complete the proof, use \cite[Theorems~5~and~16]{CaDi04} (recalled in Subsection~\ref{ssec:def}).\qed 
\end{proof} 

To prove that every cubic bent function is cubic-like bent we need the following result.

\begin{theorem}
 \label{th:deg2_bal}
Any quadratic $p$-ary function $f$ is balanced if and only if there exists $b\in\fpn$ such that the first-order derivative $D_bf(x)$ is a 
nonzero constant function. 
\end{theorem}

\begin{proof}
According to Lemma \ref{le:omega}, $f$ is balanced if and only if $\sum_{x\in\fpn}\omega^{f(x)}=0$. Rather than 
$\sum_{x\in\fpn}\omega^{f(x)}$, we consider $\big|\sum_{x\in\fpn}\omega^{f(x)}\big|^2$ and transform it as in (\ref{eq:relwd}) (with $y=0$)
\[\big|\sum_{x\in\fpn}\omega^{f(x)}\big|^2=\sum_{x,y\in\fpn}\omega^{f(y)-f(x)}=\sum_{x,b\in\fpn}\omega^{f(x+b)-f(x)}=\sum_{x,b\in\fpn}\omega^{D_bf(x)}\enspace.\]
For any $b\in\fpn$, the derivative $D_bf(x)$ is affine since $f$ is quadratic. Then $\sum_{x\in\fpn}\omega^{D_bf(x)}$ is zero unless 
$D_bf(x)$ is constant. Let $\mathcal E_f$ be the linear kernel of $f$, i.e., the set of those $b$ such that $D_bf(x)$ is constant. Then, 
$\mathcal E_f$ is an $\fp$-vector space since if $D_bf(x)$ and $D_cf(x)$ are constant then $D_{b+c}f(x)=D_bf(x+c)+D_cf(x)$ is constant and, 
therefore, $D_{sb+s'c}f(x)$ is constant for every $s,s'\in\fp$. Also, the function $b\mapsto D_bf(0)$ is $\fp$-linear over $\mathcal E_f$ 
since for every $b,c\in\mathcal E_f$, we have $D_{b+c}f(0)=D_bf(c)+D_{c}f(0)=D_bf(0)+D_{c}f(0)$. We deduce that 
\[\big|\sum_{x\in\fpn}\omega^{f(x)}\big|^2=p^n\sum_{b\in\mathcal E_f}\omega^{D_bf(0)}\]
is nonzero if and only if $b\in\mathcal E_f\mapsto D_bf(0)$ is constant (that is zero). And then by contraposition, $f$ is balanced if and 
only if there exists $b\in\mathcal E_f$ such that $D_bf(0)\neq 0$, that is, $D_bf(x)$ is a nonzero constant.\qed 
\end{proof}

\begin{theorem}
 \label{th:deg3_3like} 
Take a cubic $p$-ary function $f$ mapping $\fpn$ to $\fp$. Then $f$ is bent if and only if $f$ is cubic-like bent. 
\end{theorem}

\begin{proof}
Note that the first-order derivatives of any cubic function $f$ are quadratic, and $f$ is bent if and only if $D_af(x)$ is balanced for any 
$a\in\fpn^*$. Then, by Theorem~\ref{th:deg2_bal}, $f$ is bent if and only if for every nonzero $a$, there exists $b$ such that $D_{a,b}f$ 
is a nonzero constant function.\qed 
\end{proof} 


\section{Walsh Transform of Derivatives of $p$-ary Weakly Regular Bent Functions}
 \label{sec:deriv}
There exists an identity between the Walsh transforms of the (first-order) derivatives of a binary bent function and of its dual, which has 
initially been given in \cite[Lemma~2]{Ca99} and has interesting consequences, as reported in \cite[(6.3)~and~(6.4)]{Ca20}. In the 
following theorem, we generalize this result and its consequences to the $p$-ary weakly regular bent functions. As far as we know, this 
rather natural generalization has never been given in the literature. 

\begin{theorem}
 \label{th:bent_1deriv} 
For a prime $p$, take a $p$-ary weakly regular bent function $f:\fpn\mapsto\fp$. For any $c\in\fpn$, the Walsh transform of the first-order 
derivative of $f$ in the direction of $c$ takes on only real values and for any $b\in\fpn$, 
\begin{align}
 \label{eq:bent_1deriv1}
W_{D_c f}(b)&=W_{D_c f}(-b)=W_{D_{-c} f}(b)=W_{D_{-c}f}(-b)\\
 \label{eq:bent_1deriv2}
W_{D_c f}(b)&=\omega^{\Tr_n(bc)}W_{D_b f^*}(-c)\enspace.
\end{align} 
If $\Tr_n(bc)\neq 0$ then $W_{D_c f}(b)=0$ and if $\Tr_n(bc)=0$ then 
\[W_{D_c f}(b)=W_{D_b f^*}(c)\enspace.\]
\end{theorem}

\begin{proof}
To prove (\ref{eq:bent_1deriv2}), take any $b,c\in\fpn$ and calculate
\begin{align*}
W_{D_c f}(b)&=\sum_{x\in\fpn}\omega^{D_c f(x)-\Tr_n(bx)}\\
&=p^{-n}\sum_{x,y,v\in\fpn}\omega^{f(y)-f(x)-\Tr_n(bx+v(y-x-c))}\\
&=p^{-n}\sum_{v\in\fpn}W_f(v)\overline{W_f(v-b)}\omega^{\Tr_n(vc)}\\
&=\sum_{v\in\fpn}\omega^{f^*(v)-f^*(v-b)+\Tr_n(vc)}\\
&\stackrel{v=v+b}{=}\sum_{v\in\fpn}\omega^{D_b f^*(v)+\Tr_n((v+b)c)}\\
&=\omega^{\Tr_n(bc)}W_{D_b f^*}(-c)\enspace,
\end{align*}
where $\bar{\quad }$ denotes the complex conjugate and using that $f$ is weakly regular. 

Recall that the dual $f^*$ of a weakly regular bent function is a weakly regular bent function. Applying (\ref{eq:bent_1deriv2}) to $W_{D_b 
f^*}(-c)$ gives $W_{D_cf}(b)=W_{D_{-c}f^{**}}(-b)$. Using (\ref{eq:bent_1deriv2}) and the inverse Walsh transform (\ref{eq:W_t}) we obtain 
that 
\[\sum_{c\in\fpn}W_{D_cf}(b)\omega^{\Tr_n(cw)}=p^n\omega^{D_b f^*(-w-b)}\] 
for any $w\in\fpn$. 

Note that, for any $d\in \mathbb F_q$, we have
\begin{equation}
 \label{quad-like} 
\big(D_{c,d}f(x)\equiv\lambda\in\fp^*\big)\Rightarrow \big(\forall b\in\fpn\ \mbox{with}\ \Tr_n(bd)\neq\lambda\Rightarrow W_{D_cf}(b)=0\Big)\enspace.
\end{equation} 
Indeed, if the condition on the left-hand side is satisfied, then we have
\begin{align*}
W_{D_cf}(b)&=\sum_{x\in\fpn}\omega^{D_cf(x)-\Tr_n(bx)}=\sum_{x\in\fpn}\omega^{D_cf(x+d)-\Tr_n(b(x+d))}\\
&=\sum_{x\in\fpn}\omega^{D_cf(x)+\lambda-\Tr_n(b(x+d))}=\omega^{\lambda-\Tr_n(bd)}W_{D_cf}(b)\enspace.
\end{align*}

Note that, given two functions $f,g$ from $\fpn$ to $\fp$, we have
\[\sum_{b\in\fpn}W_f(b)\overline{W_g(b)}=\sum_{b,x,y\in\fpn}\omega^{f(x)-g(y)-\Tr_n(b(x-y))}=p^n\sum_{x\in\fpn}\omega^{f(x)-g(x)}\]
that is (if $f$ and $g$ are weakly regular bent) then
\begin{equation}
 \label{eq2}
\sum_{b\in\fpn}\omega^{f^*(b)-g^*(b)}=\sum_{x\in\fpn}\omega^{f(x)-g(x)}\enspace.
\end{equation}

\noindent(i) Taking weakly regular $f$ and $g(x)=f(x-c)+\Tr_n(ex)$, we have 
\begin{align*}
W_g(b)&=\sum_{x\in\fpn}\omega^{f(x-c)-\Tr_n((b-e)x)}\\
&=\sum_{x\in\fpn}\omega^{f(x)-\Tr_n((b-e)(x+c))}=\omega^{-\Tr_n((b-e)c)}W_f(b-e)
\end{align*}
and therefore, ($g$ is weakly regular bent and) we have $g^*(b)=f^*(b-e)-\Tr_n((b-e)c)$. Applying (\ref{eq2}) and replacing $b$ by $b+e$ 
and $x$ by $x+c$, we get 
\[\sum_{b\in\fpn}\omega^{D_e f^*(b)-\Tr_n(bc)}=\omega^{-\Tr_n(ec)}\sum_{x\in\fpn}\omega^{D_cf(x)-\Tr_n(ex)}\]
that is (changing $e$ into $b$ to compare with (\ref{eq:bent_1deriv2}))
\begin{equation}
 \label{rel1}
W_{D_cf}(b)=\omega^{\Tr_n(bc)}W_{D_b f^*}(c)\enspace.
\end{equation}
Note that (\ref{eq:bent_1deriv2}) and (\ref{rel1}) imply that for any $b$ and $c$
\begin{align}
W_{D_cf}(b)&=\omega^{2\Tr_n(bc)}W_{D_{-c}f}(b)\\
W_{D_b f^*}(c)&=W_{D_b f^*}(-c)\enspace.
\end{align}
Applying this to $f^*$ instead of $f$ (since $f^*$ is weakly regular), we get $W_{D_b f^{**}}(c)=W_{D_b f^{**}}\, (-c)$, and since 
$f^{**}(x)=f(-x)$, we have $D_b f^{**}(x)=f(-x-b)-f(-x)=D_{-b}f(-x)$ and 
\[W_{D_b f^{**}}(c)=\sum_{x\in\fpn}\omega^{D_{-b}f(-x)-\Tr_n(cx)}=\sum_{x\in\fpn}\omega^{D_{-b}f(x)+\Tr_n(cx)}=W_{D_{-b}f}(-c)\] 
and we deduce $W_{D_{-b}f}(-c)=W_{D_{-b}f}(c)$. 

\noindent(ii) Taking $g(x)=f(x)+\Tr_n(ex)$ and replacing $f(x)$ by $f(x-c)$, we have $g^*(b)=f^*(b-e)$ and $W_f(b)$ becomes
\[\sum_{x\in\fpn}\omega^{f(x-c)-\Tr_n(bx)}=\sum_{x\in\fpn}\omega^{f(x)-\Tr_n(b(x+c))}=\omega^{-\Tr_n(bc)}W_f(b)\enspace,\]
thus, $f^*(b)$ becomes $f^*(b)-\Tr_n(bc)$. Applying (\ref{eq2}) and replacing $b$ by $b+e$ and $x$ by $x+c$, we get 
\[\omega^{-\Tr_n(ec)}\sum_{b\in\fpn}\omega^{D_e f^*(b)-\Tr_n(bc)}=\omega^{-\Tr_n(ec)}\sum_{x\in\fpn}\omega^{-D_cf(x)-\Tr_n(ex)}\] 
that is 
\[\sum_{b\in\fpn}\omega^{D_e f^*(b)-\Tr_n(bc)}=\overline{\sum_{x\in\fpn}\omega^{D_cf(x)+\Tr_n(ex)}}\]
giving
\begin{equation}
 \label{rel2}
W_{D_e f^*}(c)=\overline{W_{D_c f}(-e)}\enspace.
\end{equation}
Combining (\ref{rel1}) and (\ref{rel2}) we obtain that
\begin{equation}
 \label{gen}
\overline{W_{D_c f}(-e)}=\omega^{-\Tr_n(ec)}W_{D_c f}(e)\enspace.
\end{equation}
In the binary case, (\ref{gen}) shows (as we knew already) that $W_{D_cf}(e)=0$ if $\Tr_n(ec)\neq 0$. In the odd characteristic case, 
applying (\ref{gen}) twice we obtain that 
\[W_{D_cf}(e)=\omega^{\Tr_n(ec)}\overline{W_{D_cf}(-e)}=\omega^{\Tr_n(ec)}\overline{\omega^{\Tr_n(-ec)}\overline{W_{D_cf}(e)}}=\omega^{2\Tr_n(ec)}W_{D_cf}(e)\]
that also implies $W_{D_cf}(e)=0$ when $\Tr_n(ec)\neq 0$. If $\Tr_n(ec)=0$ then by (\ref{gen}), $\overline{W_{D_c f}(-e)}=W_{D_c 
f}(e)=\overline{W_{D_c f}(e)}$ so $W_{D_c f}(e)$ is real.\qed 
\end{proof} 

Theorem~\ref{th:trinom_nwr_vec} provides an example of not weakly regular bent functions that do not satisfy identities 
(\ref{eq:bent_1deriv1}). Therefore, condition for $f$ to be weakly regular is substantial. 

The following corollary contains the conditions when adding a quadratic function to a $p$-ary weakly regular bent function gives a weakly 
regular bent function.
 
\begin{corollary}
 \label{no:plus_qu} 
Take odd $p$ and a $p$-ary weakly regular bent function $f$ and define $g(x)=f(x)+q(x)$ for some quadratic function $q$ that does not have 
any non-quadratic terms. If $q(x)\neq 0$ for any $x\in\fpn^*$ then $g$ is a bent function (we do not know if it is weakly regular). 
\end{corollary}

\begin{proof} 
Using Theorem~\ref{th:bent_1deriv}, let us try obtaining new $p$-ary bent functions by adding quadratic functions to a $p$-ary weakly 
regular bent function. Take odd $p$ and a weakly regular bent function $f$ and define $g(x)=f(x)+q(x)$ for some quadratic function 
$q(x)=\Tr_n\big(\sum_{j=0}^{n-1}a_jx^{p^j+1}\big)$ (we do not include degree $1$ or $0$ terms in the expression of $q(x)$ since they play 
no role). We have 
\begin{align*}
D_c g(x)&=D_c f(x)+\Tr_n\Big(\sum_{j=0}^{n-1}a_j(cx^{p^j}+c^{p^j}x+c^{p^j+1})\Big)\\
&=D_c f(x)+\Tr_n\Big(\sum_{j=0}^{n-1}\big((a_j^{p^{n-j}}c^{p^{n-j}}+a_jc^{p^j})x+a_jc^{p^j+1}\big)\Big)\enspace,
\end{align*}
 and therefore, 
\[W_{D_c g}(0)=\omega^{\Tr_n\big(\sum_{j=0}^{n-1}a_jc^{p^j+1}\big)}W_{D_c f}\Big(\sum_{j=0}^{n-1}\big(a_j^{p^{n-j}}c^{p^{n-j}}+a_jc^{p^j}\big)\Big)\enspace.\] 
By \cite[Theorem~5]{CaDi04}, function $g$ is bent if and only if $D_c g$ is balanced for any $c\in\fpn^*$ that is equivalent to $W_{D_c 
g}(0)=0$ according to Lemma~\ref{le:omega}. By Theorem~\ref{th:bent_1deriv}, since $f$ is a weakly regular bent function, then for any 
$b,c\in\fpn$ with $\Tr_n(bc)\neq 0$ we have $W_{D_c f}(b)=0$. Then if $\Tr_n(bc)\neq 0$ for 
$b=\sum_{j=0}^{n-1}\big(a_j^{p^{n-j}}c^{p^{n-j}}+a_jc^{p^j}\big)$ then $W_{D_c g}(0)=0$. We have
\[\Tr_n(bc)=\Tr_n\Big(c\sum_{j=0}^{n-1}\big(a_j^{p^{n-j}}c^{p^{n-j}}+a_jc^{p^j}\big)\Big)=2\Tr_n\Big(\sum_{j=0}^{n-1}a_jc^{p^j+1}\Big)=2q(c)\enspace.\]
Hence, if we have that, for any $c\in\fpn^*$, either $q(c)\neq 0$ or $W_{D_c 
f}\Big(\sum_{j=0}^{n-1}\big(a_j^{p^{n-j}}c^{p^{n-j}}+a_jc^{p^j}\big)\Big)=0$, then $g=f+q$ is a bent function. Indeed, for every $c\neq 0$, 
we have then, thanks to Theorem~\ref{th:bent_1deriv}, that $W_{D_cg}(0)=0$.\qed
\end{proof} 

Let us first see whether there exist functions $q(x)=\Tr_n\big(\sum_{j=0}^{n-1}a_jx^{p^j+1}\big)$ such that, for any $c\in\fpn^*$, we have 
$q(c)\neq 0$ (then for every weakly regular bent function $f$, we shall have that $g=f+q$ is bent, maybe not weakly regular). By the fact 
that the $(p-1)$th power of any nonzero element of $\fp$ equals $1$, this is equivalent to: $\forall x\in \fpn, (q(x))^{p-1}=1-\delta_0(x)$ 
where $\delta_0$ is the Kronecker (or Dirac) function (taking value $1$ at $0$, and value $0$ everywhere else). By the uniqueness of the 
relative trace form, this is equivalent to $\Big(\Tr_n\big(\sum_{j=0}^{n-1}a_jx^{p^j+1}\big)\Big)^{p-1}=1-\delta_0(x)\pmod{x^{p^n}-x}$. But 
calculating the left-hand side is tricky. It is easier to consider the algebraic normal form. So we choose a basis 
$(\alpha_1,\dots,\alpha_n)$ of the vector space $\fpn$ over $\fp$ and we write $x=\sum_{i=1}^n x_i\alpha_i$. Then we know, since $q$ is 
quadratic and vanishes at 0, that $q(x)$ has the form $\sum_{i=1}^n a_ix_i+\sum_{i=1}^n b_ix_i^2+\sum_{1\leq i<j\leq n} c_{i,j}x_ix_j$ 
(where the coefficients of this polynomial and the basis are chosen such that the corresponding univariate representation does not contain 
terms of degree $1$). Then $\delta_0(x)= \prod_{i=1}^n (1-x_i^{p-1})$. The condition is then 
\[\Big(\sum_{i=1}^n a_ix_i+\sum_{i=1}^n b_ix_i^2+\sum_{1\leq i<j\leq n} c_{i,j}x_ix_j\Big)^{p-1}=\sum_{\emptyset\neq I\subseteq\{1,\dots ,n\}}(-1)^{|I|+1}\prod_{i\in I}x_i^{p-1}\enspace.\]
Unfortunately, since the function on the left-hand side has an algebraic degree of at most $2(p-1)$ (see, e.g., \cite{HeKh06_1,Mei22}) and 
the function on the right-hand side has an algebraic degree equal to $n(p-1)$, such functions $q$ can exist only if $n\leq 2$ and indeed in 
such case and for $p=3$, they exist since for $n=1$ we can take $q(x)=x_1^2$ and for $n=2$ we can take $q(x)=x_1^2+x_2^2$. {\color{red}What 
about $p>3$?} For larger $n$ and $p$, the function $\sum_{i=1}^n x_i^{p-1}$ takes a nonzero value at any nonzero input if $n<p$, but it is 
not quadratic. 

We need then to find a weakly regular bent function $f$ and a function $q(x)=\Tr_n\big(\sum_{j=0}^{n-1}a_jx^{p^j+1}\big)$ such that, for 
every nonzero $c\in q^{-1}(0)$, we have $W_{D_c f}\Big(\sum_{j=0}^{n-1}\big(a_j^{p^{n-j}}c^{p^{n-j}}+a_jc^{p^j}\big)\Big)=0$ (and this may 
be easier if $q^{-1}(0)$ has a small size; note that the minimum size of $q^{-1}(0)$ equals $p^n$ minus the maximum Hamming weight of 
quadratic functions). {\color{red}TODO.} 

\begin{remark}
It seems difficult to obtain a characterization of weakly regular bent functions by the second-order derivatives, or even a property of 
such functions which would allow to show that some bent functions are not weakly regular. But some observations can be made. 

Taking any $b,c,d\in\fpn$ and a weakly regular bent function $f$ we have 
\begin{align*}
&W_{D_{c,d}f}(b)=\sum_{x\in\fpn}\omega^{D_{c,d}f(x)-\Tr_n(bx)}\\
&=p^{-3n}\sum_{x,y,z,t,v,w,s\in\fpn}\omega^{f(y)-f(z)-f(t)+f(x)-\Tr_n(bx+v(y-x-c-d)+w(z-x-c)+s(t-x-d))}\\
&=p^{-3n}\sum_{v,w,s\in\fpn}W_f(v)W_f(b-v-w-s)\overline{W_f}(-w)\overline{W_f}(-s)\omega^{\Tr_n(vc+vd+wc+sd)}\\
&=p^{-n}\sum_{v,w,s\in\fpn}\omega^{f^*(v)+f^*(b-v-w-s)-f^*(-w)-f^*(-s)+\Tr_n((v+w)c+(v+s)d)}\\
&=p^{-n}\sum_{v,w,s\in\fpn}\omega^{f^*(v)+f^*(b+w+v+s)-f^*(v+w)-f^*(v+s)-\Tr_n(wc+sd)}\enspace,
\end{align*}
where in the last identity we substitute $w=-v-w$ and $s=-v-s$.

In particular, taking $b=0$ we get 
\begin{align*}
W_{D_{c,d}f}(0)&=\sum_{x\in\fpn}\omega^{D_{c,d}f(x)}\\
&=p^{-n}\sum_{v,w,s\in\fpn}\omega^{D_{w,s}f^*(v)-\Tr_n(wc+sd)}\\
&=p^{-n}\sum_{w,s\in\fpn}W_{D_{w,s}f^*}(0)\omega^{-\Tr_n(wc+sd)}\enspace.
\end{align*}

In \cite[Theorem~5]{MeOzSi18} it is proven that $f$ is bent if and only if $\sum_{c,d\in\fpn}\omega^{D_{c,d}f(x)}\equiv p^n$ and in 
\cite[Corollary~2]{MeOzSi18}, that even weaker criterion holds: $f$ is bent if and only if 
$\sum_{c,d,x\in\fpn}\omega^{D_{c,d}f(x)}=p^{2n}$. Then 
\[\sum_{c,d\in\fpn}W_{D_{c,d}f}(b)=\sum_{x\in\fpn}\omega^{-\Tr_n(bx)}\sum_{c,d\in\fpn}\omega^{D_{c,d}f(x)}=
\left\{\begin{array}{ll}
p^{2n},&\ \mbox{if}\ b=0\\
0,&\ \mbox{otherwise}\enspace.
\end{array}\right.\]
On the other hand,
\begin{align*}
&\sum_{c,d\in\fpn}W_{D_{c,d}f}(b)=\\
&=p^{-n}\sum_{v,w,s\in\fpn}\omega^{f^*(v)+f^*(b+w+v+s)-f^*(v+w)-f^*(v+s)}\sum_{c,d\in\fpn}\omega^{-\Tr_n(wc+sd)}\\
&=p^n\sum_{v\in\fpn}\omega^{f^*(b+v)-f^*(v)}\\
&=p^nW_{D_bf^*}(0)\enspace.
\end{align*}
Thus, $W_{D_bf^*}(0)=0$ if $b\neq 0$ showing that $D_bf^*(x)$ is balanced.
\end{remark} 


\section{Trinomial Cubic Vectorial Bent Functions with Not Weakly Regular Components}
 \label{sec:vec}
In this section, we present a primary univariate construction of ternary cubic vectorial bent functions with three terms in the trace form 
and that have only not weakly regular components. Theorem~\ref{th:trinom_nwr_vec} contains the main result with the proof showing component 
functions are cubic-like bent. Identities (\ref{eq:bent_1deriv1}) for the first-order derivatives proven in Theorem~\ref{th:bent_1deriv} 
are shown not to hold for these functions that allows making the conclusion about being not weakly regular. 

In some cases, we succeed in calculating the Walsh transform coefficients and the dual of component functions. The proof that works here is 
similar to what was used in \cite[Proposition~2]{Do98} in the binary case (see also \cite{Le06}). Then it is also shown that the component 
functions are a particular MMF case of Construction~\ref{co:2}. We expect that component functions of other vectorial functions in 
Theorem~\ref{th:trinom_nwr_vec} (in addition to the case when $k$ is odd and $j\in\{0,2k\}$) are in MMF class as well. The following lemma 
is needed in Theorem~\ref{th:trinom_nwr_vec}. 

\begin{lemma}
 \label{le:vec}
Let $n=4k$, take $j\not\equiv k\pmod 2$ and $t$ odd. Also, take any $\alpha\in\fthreek^*$, $c\in\fthreen^*$, and let 
$b=\zeta^{t(3^k+1)/2}\in\F_{3^{2k}}^*$, where $\zeta$ is a primitive element of $\F_{3^{2k}}$. If $c\in\F_{3^{2k}}^*$ or $\Tr^n_k(bc)\neq 
0$ then 
\begin{equation}
 \label{eq:le2}
\Tr_n\big((\alpha bc^{3^j}+\alpha^{3^{-j}}b^{3^{-j}}c^{3^{-j}})d\big)
\end{equation}
is not identically zero as a function of $d\in\fthreen$. Moreover, for every $c\in\F_{3^{2k}}^*$ there exists $d\in\F_{3^{2k}}^*$ such that 
(\ref{eq:le2}) is not zero and for every $c\in\fthreen^*$ with $\Tr^n_k(bc)\neq 0$ there exists $d\in\fthreek^*$ such that (\ref{eq:le2}) 
is not zero. 
\end{lemma}

\begin{proof} 

Assume $\alpha=\zeta^{e(3^k+1)}$ for some $e\geq 0$. Firstly, if $c\in\F_{3^{2k}}^*$ then also assume $d\in\F_{3^{2k}}^*$ to get
\begin{equation}
 \label{eq:le2_1}
\Tr_n\big((\alpha bc^{3^j}+\alpha^{3^{-j}}b^{3^{-j}}c^{3^{-j}})d\big)=-\Tr_{2k}\big((\alpha bc^{3^j}+\alpha^{3^{-j}}b^{3^{-j}}c^{3^{-j}})d\big)\enspace.
\end{equation}
By the uniqueness of the relative trace form (\ref{eq:le2_1}) ($o(1)=2k$ for modulus $3^{2k}-1$), it is identically zero for all 
$d\in\F_{3^{2k}}$ if and only if $\alpha^{3^j-1}b^{3^j-1}c^{3^{2j}-1}=-1$, since $\alpha$, $b$ and $c$ are nonzero. Solve the latter 
equation for the unknown $c\in\F_{3^{2k}}^*$ that is equivalent to checking if the following modular equation has an integer solution $z$ 
\[e(3^k+1)(3^j-1)+t(3^k+1)(3^j-1)/2+z(3^{2j}-1)\equiv(3^{2k}-1)/2\pmod{3^{2k}-1}\enspace.\]
Obviously, $\gcd(3^{2j}-1,3^{2k}-1)=9^{\gcd(j,k)}-1$ is divisible by $8$ and for this equivalence to have a solution, we need that 
$9^{\gcd(j,k)}-1$ divides $(3^k+1)(3^k-1-t(3^j-1)-2e(3^j-1))/2$. Note that $(3^k-1-t(3^j-1))/2$ is odd since only one of $3^k-1$ and 
$3^j-1$ is divisible by four and $t$ is odd. Thus, $(3^k-1-t(3^j-1)-2e(3^j-1))/2$ is odd. Also, $3^k+1$ is not divisible by $8$ so no 
solution for $z$ exists. Thus, for every $c\in\F_{3^{2k}}^*$ there exists $d\in\F_{3^{2k}}^*$ such that (\ref{eq:le2}) is not zero. 

Secondly, if $\Tr^n_k(bc)\neq 0$ then assume $d\in\F_{3^k}^*$ to get
\begin{equation}
 \label{eq:le2_2}
\Tr_n\big((\alpha bc^{3^j}+\alpha^{3^{-j}}b^{3^{-j}}c^{3^{-j}})d\big)=\Tr_k\big(\Tr^n_k(\alpha bc^{3^j}+\alpha^{3^{-j}}b^{3^{-j}}c^{3^{-j}})d\big)\enspace.
\end{equation}
By the uniqueness of the relative trace form (\ref{eq:le2_2}) ($o(1)=k$ for modulus $3^k-1$), it is identically zero for all $d\in\fthreek$ 
if and only if 
\[\Tr^n_k(\alpha^{3^j}b^{3^j}c^{3^{2j}}+\alpha bc)=\alpha^{3^j}b^{-3^j(3^j-1)}\Tr^n_k(bc)^{3^{2j}}+\alpha\Tr^n_k(bc)=0\]
that is equivalent to $\Tr^n_k(bc)=0$ or $\alpha^{3^j-1}b^{-3^j(3^j-1)}\Tr^n_k(bc)^{3^{2j}-1}=-1$. Since $\Tr^n_k(bc)\neq 0$, we only need 
to check if the latter identity can hold and, in particular, this will require the following modular equation to have some integer solution 
$z$ 
\[e(3^j-1)-t3^j(3^j-1)/2+z(3^{2j}-1)\equiv(3^k-1)/2\pmod{3^k-1}\enspace.\]
Observe that $3^{\gcd(j,k)}-1$ divides $\gcd(3^{2j}-1,3^k-1)$ and for this equivalence to have a solution, we need that $3^{\gcd(j,k)}-1$ 
divides $(3^k-1+t3^j(3^j-1)-2e(3^j-1))/2$ which is odd since only one of $3^k-1$ and $3^j-1$ is divisible by four and $t$ is odd (so no 
solution for $z$ exists). Thus, for every $c\in\fthreen^*$ with $\Tr^n_k(bc)\neq 0$ there exists $d\in\fthreek^*$ such that (\ref{eq:le2}) 
is not zero. 

In both cases, the parity conditions on $k$, $j$ and $t$ play role.\qed 
\end{proof} 

\begin{theorem}
 \label{th:trinom_nwr_vec}
Let $n=4k$, take $j\not\equiv k\pmod 2$ and $t$ odd. Then ternary vectorial function $F:\fthreen\mapsto\fthreek$ given by 
\begin{equation}
 \label{eq:trinom_nwr_vec} 
F(x)=\Tr^n_k\big(x^{3^k+2}-x^{2\cdot 3^k+1}+\zeta^{t(3^k+1)/2} x^{3^j+1}\big)\enspace,
\end{equation} 
where $\zeta$ is a primitive element of $\F_{3^{2k}}$, is a cubic bent function. Moreover, component functions of $F$ are all not weakly 
regular MMF bent functions. In particular, when $k$ is odd and $j\in\{0,2k\}$ then the component bent function $\Tr_k(F(x))$ is covered (up 
to EA equivalence) by the particular MMF form (\ref{eq:co2_sp}) of Construction~\ref{co:2} with $d=k$ and $\pi()$ identical mapping, and it 
is dual bent. 
\end{theorem}

\begin{proof} 
Denoting $b=\zeta^{t(3^k+1)/2}$ we have that $b^{3^k}=-b$ (since 
\[3^k t(3^k+1)/2=t(3^{2k}-1)/2+t(3^k+1)/2\]
and $t$ is odd) and then $\Tr^n_k(b)=0$. 

To prove bentness, we show that all component functions of $F$ are bent by analyzing their second-order derivatives in the direction of 
elements $c,d\in\fthreen^*$. For any $\alpha\in\fthreek^*$, denote $f_{\alpha}(x)=\Tr_k(\alpha\Tr^n_k(F(x)))$ a component function of $F$. 
Begin with the first-order derivative 
\begin{align}
 \label{eq:nwr2_1d_vec}
\nonumber&D_c f_{\alpha}(x)=\Tr_k\big(\alpha\Tr^n_k(D_cF(x))\big)\\
\nonumber&=\Tr_k\Big(\alpha\Tr^n_k\big(-c^{2\cdot 3^k+1}+c^{3^k+2}+bc^{3^j+1}+bcx^{3^j}+bc^{3^j}x-(c^{2\cdot 3^k}+c^{3^k+1})x\\
&\ +(c^2+c^{3^k+1})x^{3^k}+(c^{3^k}-c)x^{3^k+1}+c^{3^k}x^2-cx^{2\cdot 3^k}\big)\Big)
\end{align}
and then 
\begin{align*}
&D_{c,d}f_{\alpha}(x)=\Tr_k\Big(\alpha\Tr^n_k\big((c^{3^k}d+cd^{3^k}-cd)x^{3^k}+(c^{3^k}d^{3^k}-c^{3^k}d-cd^{3^k})x\\
&\ -c^{2\cdot 3^k}d-cd^{2\cdot 3^k}+c^{3^k+1}d^{3^k}+c^{3^k}d^{3^k+1}-c^{3^k+1}d-cd^{3^k+1}+c^{3^k}d^2+c^2 d^{3^k}\\
&\ +b(c^{3^j}d+cd^{3^j})\big)\Big)\\
&=\Tr_k\Big(\alpha\Tr^n_k\big((cd^{3^{3k}}+c^{3^{3k}}d-c^{3^{3k}}d^{3^{3k}}+c^{3^k}d^{3^k}-c^{3^k}d-cd^{3^k})x\\
&\ +(c^{3^k}-c)d^{3^k+1}-(c^{3^{2k}}-c)^{3^k}d^2+(c^{3^{3k}+1}+c^{2\cdot 3^{3k}}-c^{2\cdot 3^k}-c^{3^k+1})d\\
&\ +b(c^{3^j}d+cd^{3^j})\big)\Big)\\
&=\Tr_k\Big(\alpha\Tr^n_k\big(L_c(d)x+(c^{3^k}-c)d^{3^k+1}-(c^{3^{2k}}-c)^{3^k}d^2\\
&\ +(c^{3^{2k}}-c)^{3^k}(c^{3^{3k}}+c^{3^k}+c)d\big)\Big)+\Tr_n\big((\alpha bc^{3^j}+\alpha^{3^{-j}}b^{3^{-j}}c^{3^{-j}})d\big)\enspace,
\end{align*}
where the linearized polynomial
\[L_c(d)=(c^{3^k}-c)^{3^{3k}}d^{3^{3k}}+(c^{3^k}-c)d^{3^k}+(c^{3^{2k}}-c)^{3^k}d\enspace.\]
Note that for every $s\in\fthreek$,
\[D_{c,d+s}f_{\alpha}(x)=D_{c,d}f_{\alpha}(x)+D_{c,s}f_{\alpha}(x)\enspace.\]

If $c\in\fthreek^*$ then $L_c(d)\equiv 0$ and
\[D_{c,d}f_{\alpha}(x)=\Tr_n\big((\alpha bc^{3^j}+\alpha^{3^{-j}}b^{3^{-j}}c^{3^{-j}})d\big)\enspace.\]
By Lemma~\ref{le:vec}, for every $c\in\fthreek^*$ there exists $d\in\F_{3^{2k}}^*$ such that $D_{c,d}f_{\alpha}$ is a nonzero constant 
function. If $c\in\F_{3^{2k}}^*$ then also take $d\in\F_{3^{2k}}^*$ and observing that $L_c(d)=0$ we obtain 
\begin{align*}
D_{c,d}f_{\alpha}(x)&=\Tr_k\big(\alpha\Tr^n_k\big((c^{3^k}-c)d^{3^k+1}\big)\big)+\Tr_n\big((\alpha bc^{3^j}+\alpha^{3^{-j}}b^{3^{-j}}c^{3^{-j}})d\big)\\
&=\Tr_n\big((\alpha bc^{3^j}+\alpha^{3^{-j}}b^{3^{-j}}c^{3^{-j}})d\big)\enspace.
\end{align*}
By Lemma~\ref{le:vec}, for every $c\in\F_{3^{2k}}^*$ there exists $d\in\F_{3^{2k}}^*$ such that $D_{c,d}f_{\alpha}$ is a nonzero constant 
function. Further assume $c\notin\fthreek$. 

Take $c\in\fthreen$ with $\Tr^n_k(bc)=0$ (note this is equivalent to $c+c^{3^{2k}}\in\fthreek$ and to $(c-c^{3^k})^{3^{2k}}=-(c-c^{3^k})$) 
and this holds if and only if $c=b^{-1}(l-l^{3^k})$ for some $l\in\fthreen$. Require additionally that $c\notin\fthreek$ or, equivalently, 
$l\notin\F_{3^{2k}}$. Then 
\[bL_c(d)=(l^{3^{3k}}-l^{3^k})d^{3^{3k}}+(l^{3^{2k}}-l)d^{3^k}-(l^{3^{3k}}+l^{3^{2k}}-l^{3^k}-l)d\]
and it can be checked directly that, in addition to $\fthreek$, other root of $L_c(d)$ is $l^{3^{3k}}+l$. This means that for such $c$, we 
know $3^{2k}$ roots of $L_c(d)$ belonging to the $2$-dimensional $\fthreek$-linear vector space expanded by the basis $\{1,l^{3^{3k}}+l\}$. 
Further, take $d=(l^{3^{3k}}+l)s$ with $s\in\fthreek^*$ to obtain 
\begin{align*}
&D_{c,d}f_{\alpha}(x)=\Tr_k\Big(\alpha\Tr^n_k\Big(\big((c^{3^k}-c)(l^{3^{3k}}+l)^{3^k}-(c^{3^{2k}}-c)^{3^k}(l^{3^{3k}}+l)\big)(l^{3^{3k}}+l)s^2\\
&\ +(c^{3^{2k}}-c)^{3^k+1}(l^{3^{3k}}+l)s\Big)\Big)+\Tr_n\big((\alpha bc^{3^j}+\alpha^{3^{-j}}b^{3^{-j}}c^{3^{-j}})(l^{3^{3k}}+l)s\big)\\
&=\Tr_k\Big(\alpha\Tr^n_k\Big(\big((c^{3^k}-c)(l^{3^k}+l)-(c^{3^{2k}}-c)^{3^k}(l^{3^{3k}}+l)\\
&\ +(c^{3^k}-c)^{3^k}(l^{3^k}+l)^{3^k}+(c^{3^{2k}}-c)(l^{3^k}+l)\big)ls^2\Big)\Big)+\Tr_n(Ks)\\
&=\Tr_k\Big(\alpha\Tr^n_k\Big((c^{3^{2k}}+c^{3^k}+c)(l^{3^k}+l)ls^2\\
&\ +\big((c^{3^k}-c)^{3^k}(l^{3^k}+l)^{3^k}-(c^{3^{2k}}+c^{3^k}+c)(l^{3^{3k}}+l)\big)ls^2\Big)\Big)+\Tr_n(Ks)\\
&=\Tr_k\Big(\alpha\Tr^n_k\Big(\big((c^{3^{2k}}-c)^{3^k}(l-l^{3^{2k}})^{3^k}+(c^{3^k}-c)^{3^k}(l^{3^k}+l)^{3^k}\big)ls^2\Big)\Big)+\Tr_n(Ks)\\
&=\Tr_k\Big(\alpha\Tr^n_k\Big(\big((c^{3^{2k}}-c)(l-l^{3^{2k}})+(c^{3^k}-c)(l^{3^k}+l)\big)l^{3^{3k}}s^2\Big)\Big)+\Tr_n(Ks)\\
&=\Tr_k\Big(\alpha\Tr^n_k\Big(\big((c^{3^{2k}}+c^{3^k}+c)l^{3^{3k}+1}-(c^{3^{2k}}-c)l^{3^{3k}+3^{2k}}+(c^{3^k}-c)l^{3^{3k}+3^k}\big)s^2\Big)\Big)+\Tr_n(Ks)\\
&=\Tr_k\Big(\alpha\Tr^n_k\Big(\big((c^{3^{2k}}+c^{3^k}+c)l^{3^{3k}+1}-(c^{3^{3k}}-c^{3^k})l^{3^{3k}+1}+(c^{3^k}-c)l^{3^k(3^{2k}+1)}\big)s^2\Big)\Big)+\Tr_n(Ks)\\
&=\Tr_k\Big(\alpha\Tr^n_k\Big(\big((c^{3^{2k}}-c^{3^k}+c-c^{3^{3k}})l^{3^{3k}+1}+(c^{3^k}-c)l^{3^k(3^{2k}+1)}\big)s^2\Big)\Big)+\Tr_n(Ks)\\
&=\Tr_{2k}\big(\alpha l^{3^k(3^{2k}+1)}s^2\Tr^n_{2k}(c^{3^k}-c)\big)+\Tr_k\big(\Tr^n_k(K)s\big)\\
&=\Tr_k\big(\Tr^n_k(K)s\big)\enspace,
\end{align*}
where $K$ is defined explicitly and $c^{3^{3k}}+c^{3^k}=c^{3^{2k}}+c$ is used a number of times. By the uniqueness of the relative trace 
form $\Tr_k\big(\Tr^n_k(K)s\big)$ ($o(1)=k$ for modulus $3^k-1$), it is identically zero for all $s\in\fthreek^*$ if and only if 
$\Tr^n_k(K)=0$. Use $c=b^{-1}(l-l^{3^k})$ to obtain 
\begin{align*}
\Tr^n_k(K)&=\Tr^n_k\Big(\big(\alpha b^{3^n-3^j}(l-l^{3^k})^{3^j}+\alpha^{3^{-j}}(l-l^{3^k})^{3^{-j}}\big)(l^{3^{3k}}+l)\Big)\\
&=\Tr^n_k\big((\alpha b^{3^n-3^j}l^{3^j}+\alpha^{3^{-j}}l^{3^{-j}})(l-l^{3^{2k}})\big)\\
&=\Tr^{2k}_k\big(\alpha b^{3^n-3^j}(l-l^{3^{2k}})^{3^j+1}+\alpha^{3^{-j}}(l-l^{3^{2k}})^{3^{-j}+1}\big)\\
&=\Tr^{2k}_k\big(\alpha b^{3^n-3^j}r^{-3^j-1}v^{3^j+1}+\alpha^{3^{-j}}r^{-3^{-j}-1}v^{3^{-j}+1}\big)\\
&=\Tr^{2k}_k\big(\alpha r^{t(3^k+1)(3^n-3^j)-3^j-1}v^{3^j+1}+\alpha^{3^{-j}}r^{-3^{-j}-1}v^{3^{-j}+1}\big)\\
&=\Tr^{2k}_k\big(\alpha Rr^{-(3^j+1)}v^{3^j+1}+\alpha^{3^{-j}}r^{-(3^{-j}+1)}v^{3^{-j}+1}\big)\enspace,
\end{align*}
where $r=\xi^{(3^{2k}+1)/2}$ and $\xi$ is a primitive element of $\fthreen$ (so $r^{3^{2k}}=-r$ and $b=r^{t(3^k+1)}$), 
$v=\Tr^n_{2k}(rl)=r(l-l^{3^{2k}})\neq 0$ (since $l\notin\F_{3^{2k}}$) and $R=r^{t(3^k+1)(3^n-3^j)}=\eta^{t(3^n-3^j)/2}\in\fthreek^*$ with 
$\eta$ a primitive element of $\fthreek$. Further, $\Tr^n_k(K)=0$ if and only if 
\begin{align*}
&\Tr^{2k}_k\big(\alpha^{3^j}R^{3^j}r^{-3^j(3^j+1)}v^{3^j(3^j+1)}+\alpha r^{-(3^j+1)}v^{3^j+1}\big)=\\
&=\Tr^{2k}_k\big(\alpha^{3^j}R^{3^j}w^{3^j}+\alpha w\big)=\alpha^{3^j}R^{3^j}\Tr^{2k}_k(w)^{3^j}+\alpha\Tr^{2k}_k(w)=0\enspace,
\end{align*}
where we denote $w=r^{-(3^j+1)}v^{3^j+1}\in\F_{3^{2k}}^*$. This holds if and only if $\Tr^{2k}_k(w)=0$ or 
$\alpha^{3^j-1}R^{3^j}\Tr^{2k}_k(w)^{3^j-1}=-1$. 

Assume $\alpha=\eta^e$ for some $e\geq 0$. For the latter identity to hold, in particular, we need the following modular equation to have 
an integer solution $z$ 
\[e(3^j-1)+t3^j(3^n-3^j)/2+z(3^j-1)\equiv(3^k-1)/2\pmod{3^k-1}\enspace.\]
Since $\gcd(3^j-1,3^k-1)=3^{\gcd(j,k)}-1$ is even, then for this equivalence to have a solution we need $(3^k-1-t3^j(3^n-3^j)-2e(3^j-1))/2$ 
to be even, but it is odd since only one of $3^k-1$ and $3^n-3^j$ is divisible by four and $t$ is odd (so no solution for $z$ exists). Here 
the parity conditions on $k$, $j$ and $t$ play role. 

It remains to prove that $\Tr^{2k}_k(w)\neq 0$. Since $r$ and $v$ (and $w$) are nonzero, the opposite is equivalent to 
\[w^{3^k-1}=r^{-(3^k-1)(3^j+1)}v^{(3^k-1)(3^j+1)}=-1\enspace.\]
For the latter identity to hold, in particular, we need the following modular equation to have an integer solution $z$ 
\[-(3^k-1)(3^j+1)/2+z(3^k-1)(3^j+1)\equiv(3^{2k}-1)/2\pmod{3^{2k}-1}\enspace.\]
Since 
\[\gcd((3^k-1)(3^j+1),3^{2k}-1)=(3^k-1)\gcd(3^j+1,3^k+1)\]
and $\gcd(3^j+1,3^k+1)$ is even, then for this equivalence to have a solution we need $(3^k+1+3^j+1)/2$ to be even, but it is odd since 
only one of $3^k+1$ and $3^j+1$ is divisible by four (so no solution for $z$ exists). Here the parity conditions on $k$ and $j$ play role. 


Conclude that if $\Tr^n_k(bc)=0$ and $c\notin\fthreek$, then $D_{c,d}f_{\alpha}$ is nonzero for some $d=(l^{3^{3k}}+l)s$ with 
$s\in\fthree^*k$. 

Finally, assume $\Tr^n_k(bc)\neq 0$ (then $c\notin\fthreek$) and also take $d\in\F_{3^k}^*$ that gives 
\[D_{c,d}f_{\alpha}(x)=\Tr_n\big((\alpha bc^{3^j}+\alpha^{3^{-j}}b^{3^{-j}}c^{3^{-j}})d\big)\enspace.\]
By Lemma~\ref{le:vec}, for every such $c$, there exists $d\in\fthreek^*$ such that $D_{c,d}f$ is a nonzero constant function. 

Combining all particular cases, conclude that for every $c\in\fthreen^*$ there exists $d\in\fthreen^*$ such that $D_{c,d}f_{\alpha}(x)$ is 
a nonzero constant defined by $c$ and $d$ that does not depend on $x$. Thus, $f_{\alpha}$ is cubic-like bent and it remains to apply 
Theorem~\ref{th:3like_bent} to conclude that $f_{\alpha}$ is bent. Since this holds for any $\alpha\in\fthreek^*$, we conclude that $F$ is 
a cubic vectorial bent function. 

To check if $f_{\alpha}$ is weakly regular, take the first-order derivative of $f_{\alpha}$ in (\ref{eq:nwr2_1d_vec}). Additionally 
assuming $c\in\fthreek^*$, we obtain 
\begin{align*}
D_{c}f_{\alpha}(x)&=\Tr_k\big(\alpha\Tr^n_k\big(bc^{3^j+1}+bcx^{3^j}+bc^{3^j}x\big)\big)\\
&=\Tr_n\big(\alpha bc^{3^j+1}+(\alpha bc^{3^j}+\alpha^{3^{-j}}b^{3^{-j}}c^{3^{-j}})x\big)\enspace.
\end{align*}
Then the Walsh transform of $D_{c}f_{\alpha}$ in point $q\in\fthreen$ is
\[W_{D_{c}f_{\alpha}}(q)=\left\{\begin{array}{ll}
3^n\omega^{\Tr_n(\alpha bc^{3^j+1})},&\ \mbox{if}\ q=\alpha bc^{3^j}+\alpha^{3^{-j}}b^{3^{-j}}c^{3^{-j}}\\
0,&\ \mbox{otherwise .}
\end{array}\right.\]
By Theorem~\ref{th:bent_1deriv}, if $f_{\alpha}$ is weakly regular then $W_{D_{c}f_{\alpha}}(q)=W_{D_{c}f_{\alpha}}(-q)$. In our case, this 
is possible only if $\alpha bc^{3^j}+\alpha^{3^{-j}}b^{3^{-j}}c^{3^{-j}}=0$ and this does not hold by Lemma~\ref{le:vec}. Thus, 
$f_{\alpha}$ is a not weakly regular bent function.\qed 
\end{proof} 

Now consider particular component function $f(x)=\Tr_k(F(x))$ of vectorial bent function $F$ in (\ref{eq:trinom_nwr_vec}) with $k$ odd and 
$j\in\{0,2k\}$. Since polynomial $x^4+x-1$ is irreducible over $\fthree$ it is also irreducible over $\fthreek$ (see 
\cite[Lemma~?]{LiNi97}). Taking $a\in\F_{3^4}$ such that $a^4+a=1$ we obtain that $\fthreen=\fthreek[a]$. Note that $a$ is a primitive 
element of $\F_{3^4}$. In particular, any $x\in\fthreen$ can be uniquely represented as 
\begin{equation}
 \label{eq:x_exp1}
x=x_3 a^3+x_2 a^2+x_1 a+x_0
\end{equation}
with $x_i\in\fthreek$. Using $a^4=1-a$ calculate that
\begin{align*}
\Tr^n_k(a)&=a+a^{3^k}+a^{3^{2k}}+a^{3^{3k}}=a+a^3+a^9+a^{27}=0=\Tr^n_k(a^3)\\
\Tr^n_k(a^2)&=a^2+a^6+a^{18}-a^{14}=0\enspace.
\end{align*}
Also $\Tr^n_k(1)=1$ and note that 
\begin{align*}
&\Tr^n_{2k}(a)=\Tr^n_{2k}(a^2)=\Tr^n_{2k}(a^9)=\Tr^n_{2k}(a^{18})=a^3+a^2-a=-a^{20}=-I\\
&\Tr^n_{2k}(a^3)=\Tr^n_{2k}(a^6)=\Tr^n_{2k}(a^{27})=\Tr^n_{2k}(a^{54})=-a^3-a^2+a=a^{20}=I\enspace,
\end{align*}
where $I$ is a primitive $4$th root of unity in $\fthreen$ (and $I\in\F_{3^2}$). 

One of the component bent functions of $F$ in (\ref{eq:trinom_nwr_vec}) taken with $j=2k$ and $t=(3^k-1)/2$ (which is odd for odd $k$) is 
\begin{equation}
 \label{eq:trinom_nwr1I} 
f(x)=\Tr_k(F(x))=\Tr_n\big(x^{3^k+2}-x^{2\cdot 3^k+1}+I x^{3^{2k}+1}\big)\enspace.
\end{equation}
Since $x^{3^{2k}+1}\in\F_{3^{2k}}$, the above trace identities mean that in (\ref{eq:trinom_nwr_vec}) with $j=2k$, the coefficient at the 
quadratic term can also be taken equal to $a$ or $a^2$ for every $a$ satisfying $a^4+a=1$ and the resulting component bent functions are 
either equal or EA equivalent to (\ref{eq:trinom_nwr1I}) (substitution $x=-x$). 

Also
\begin{equation}
 \label{eq:2univ}
\Tr^n_k(x)=x_0,\ \Tr^n_k(ax)=x_3,\ \Tr^n_k(a^2 x)=x_2,\ \Tr^n_k(a^3 x)=x_1
\end{equation}
that can be used to get back to the univariate representation.

Denoting $b=\zeta^{t(3^k+1)/2}$ it was already noted that $b^{3^k}=-b$. Obviously, $\Tr^n_k(b)=0$ and calculate
\begin{equation}
 \label{eq:Tr_ba}
\Tr^n_k(b a)=\Tr^n_k(b a^2)=-ba(a^2+a-1)=ba^{20}=-\Tr^n_k(b a^3)\enspace.
\end{equation}

Using (\ref{eq:x_exp1}), convert $f(x)$ in its univariate representation into the polynomial of four variables over $\fthreek$. Note that 
$3^{2k}\equiv 9\pmod{80}$ for odd $k$ and if $k\equiv 1\pmod{4}$ (resp. $k\equiv 3\pmod{4}$) then $3^k\equiv 3\pmod{80}$ (resp. $3^k\equiv 
27\pmod{80}$). First, take $\Tr_n\big(x^{3^k+1}(x-x^{3^k})\big)$ that in $4$ variables becomes, when $k\equiv 1\pmod{4}$, 
\begin{align*}
&\Tr_n\big((x_3 a^9+x_2 a^6+x_1 a^3+x_0)(x_3 a^3+x_2 a^2+x_1 a+x_0)\\
&\ \times(x_3 a^3+x_2 a^2+x_1 a+x_0-x_3 a^9-x_2 a^6-x_1 a^3-x_0)\big)
\end{align*}
and, when $k\equiv 3\pmod{4}$,
\begin{align*}
&\Tr_n\big((x_3 a+x_2 a^{54}+x_1 a^{27}+x_0)(x_3 a^3+x_2 a^2+x_1 a+x_0)\\
&\ \times(x_3 a^3+x_2 a^2+x_1 a+x_0-x_3 a-x_2 a^{54}-x_1 a^{27}-x_0)\big)\enspace.
\end{align*}
After multiplication, reduction using $a^4=1-a$, knowing that $\Tr^n_k(a^i)=0$ for $i=1,2,3$ and using (\ref{eq:Tr_ba}), we obtain for any 
odd $k$, 
\[\Tr_n\big(x^{3^k+1}(x-x^{3^k})\big)=\Tr_k\big((-1)^{(k-1)/2}(x_1+x_1^2 x_2+x_1 x_2^2+x_2+x_1^2 x_3+x_2 x_3^2+x_3)\big)\enspace.\]

Now take the remaining quadratic term. It $j=0$ then
\[\Tr_n(bx^2)=\Tr_k\big(ba^{20}(x_1^2+x_1 x_2-x_2^2+x_1 x_3-x_3^2-(x_1+x_2-x_3)x_0)\big)\]
and if $j=2k$ then
\begin{align*}
\Tr_n(bx^{3^{2k}+1})&=\Tr_n\big(b(x_3 a^{27}+x_2 a^{18}+x_1 a^9+x_0)(x_3 a^3+x_2 a^2+x_1 a+x_0)\big)\\
&=\Tr_k\big(-b a^{20}(x_1^2+x_1 x_2-x_2^2+x_1 x_3-x_3^2+(x_1+x_2-x_3)x_0)\big)\enspace.
\end{align*}

Adding the two expressions, we obtain
\begin{align*}
f(x)&=\Tr_k\big((-1)^{(k-1)/2}(x_1+x_1^2 x_2+x_1 x_2^2+x_2+x_1^2 x_3+x_2 x_3^2+x_3)\\
&\ +(-1)^{j/2}b a^{20}(x_1^2+x_1 x_2-x_2^2+x_1 x_3-x_3^2)-b a^{20}(x_1+x_2-x_3)x_0\big)\\
&=\Tr_k\big(g(x_1,x_2,x_3)-b a^{20}(x_1+x_2-x_3)x_0\big)\\
&=\Tr_k\big(g(x_1,x_2,x_1+x_2-b^{-1} a^{20} x_3)+x_0 x_3\big)\\
&=\Tr_k\big(\tilde{g}(x_1,x_2,x_3)+x_0 x_3\big)\enspace,
\end{align*}
where we substitute $x_3=-b a^{20}(x_1+x_2-x_3)$ and, using $s_k=(-1)^{(k-1)/2}$ and $s_j=(-1)^{j/2}$ for the signs,
\begin{align*}
&\tilde{g}(x_1,x_2,x_3)=B\big(s_k B x_2+s_j\big)x_3^2\\
&\ +s_k B(x_2^2-x_1^2+x_1 x_2-1)x_3+s_j(x_2-x_1)x_3\\
&\ +s_k(x_1^3+x_2^3-x_1-x_2)-s_j B^{-1}(x_1^2+x_2^2)
\end{align*}
denoting $B=b^{-1}a^{20}\in\fthreek^*$, observe that $B^2=-b^{-2}$ and $B^{-1}=-b a^{20}=-b^2 B$. Note that 
$\Tr_k\big(\tilde{g}(x_1,x_2,x_3)+x_0 x_3\big)$ is the particular form (\ref{eq:co2_sp}) of Construction~\ref{co:2} with $d=k$ and $\pi()$ 
identical mapping (here we use the similarity between $\fthree^k$ and $\fthreek$, also for the conventional dot product of $\fthree^k$ take 
the standard inner product for $\fthreek$ that is $\Tr_k()$). Therefore, it suffices to show that 
$\Tr_k(\tilde{g}_{y_0}(x_1,x_2))=\Tr_k(\tilde{g}(x_1,x_2,y_0))$ is bent for any $y_0\in\fthreek$. 

Remarkable that for $t=(3^k-1)/2$ the underlying polynomial over $\fthreek$ does not depend on $k$ meaning that our bent function belongs 
the rare class of exceptional polynomials. 


Take any $y\in\fthreen$ expanded as $y=y_3 a^3+y_2 a^2+y_1 a+y_0$ with $y_i\in\fthreek$, then $\Tr_n(xy)=\Tr_k(y_0 x_0+y_3 x_1+y_2 x_2+y_1 
x_3)$, and compute the Walsh transform of $f$ in point $y$. 
\begin{align*}
W_f(y)&=\sum_{x\in\fthreen}\omega^{f(x)-\Tr_n(xy)}\\
&=\sum_{x_0,x_1,x_2,x_3\in\fthreek}\omega^{\Tr_k\big(\tilde{g}(x_1,x_2,x_3)+(x_3-y_0)x_0-y_3 x_1-y_2 x_2-y_1 x_3\big)}\\
&=3^k\omega^{-\Tr_k(y_0 y_1)}\sum_{x_1,x_2\in\fthreek}\omega^{\Tr_k\big(\tilde{g}(x_1,x_2,y_0)-y_3 x_1-y_2 x_2\big)}\\
&=3^k\omega^{-\Tr_k(y_0 y_1)}W_{\tilde{g}_{y_0}}(y_3,y_2)\enspace,
\end{align*}
where $\tilde{g}_{y_0}(x_1,x_2)=\tilde{g}(x_1,x_2,y_0)$ and functions $\Tr_k(\tilde{g}_{y_0}(x_1,x_2))$ are of degree $2$ both in $x_1$ and 
$x_2$. By \cite[Proposition~1]{HeKh06_1}, all quadratic $p$-ary bent functions are (weakly) regular, i.e., dual bent. Therefore, as 
explained in Construction~\ref{co:2}, if we prove that $\Tr_k(\tilde{g}_{y_0})$ are bent then $f$ is proven to be dual bent. 

The Walsh transform of $\Tr_k(\tilde{g}_{y_0})$ can be calculated using (\ref{eq:quad_Wtc}) as
\begin{align*}
&W_{\Tr_k(\tilde{g}_{y_0})}(y_3,y_2)=\omega^{\Tr_k\big(B y_0(s_j y_0-s_k)\big)}\\
&\ \times\sum_{x_1,x_2\in\fthreek}\omega^{\Tr_k\big(s_k B^2 x_2 y_0^2+s_k B(x_2^2-x_1^2+x_1 x_2)y_0+s_j(x_2-x_1)y_0-s_j B^{-1}(x_1^2+x_2^2)-y_3 x_1-y_2 x_2\big)}\\
&=\omega^{\Tr_k\big(B y_0(s_j y_0-s_k)\big)}\\
&\ \times\sum_{x_1,x_2\in\fthreek}\omega^{\Tr_k\big((s_k B y_0-s_j B^{-1})x_2^2-(s_k B y_0+s_j B^{-1})x_1^2\big)}\\
&\ \times\omega^{\Tr_k\big((s_k B^2 y_0^2+s_j y_0-y_2)x_2-(s_j y_0+y_3)x_1+s_k B y_0 x_1 x_2\big)}\\
&=\eta(s_k B y_0-s_j B^{-1})i^k 3^{k/2}\omega^{\Tr_k\big(B y_0(s_j y_0-s_k)\big)}\\
&\ \times\sum_{x_1\in\fthreek}\omega^{-\Tr_k\big((s_k B y_0+s_j B^{-1})x_1^2+(s_j y_0+y_3)x_1+(s_k B^2 y_0^2+s_j y_0-y_2+s_k B y_0 x_1)^2(s_k B y_0-s_j B^{-1})^{-1}\big)}\\
&=\eta(s_k B y_0-s_j B^{-1})i^k 3^{k/2}\omega^{\Tr_k\big(B y_0(s_j y_0-s_k)+(s_k B^2 y_0^2+s_j y_0-y_2)^2(s_k B y_0-s_j B^{-1})^{-1}\big)}\\
&\ \times\sum_{x_1\in\fthreek}\omega^{\Tr_k\big(\big((B^{-2}+B^2 y^2_0)x_1^2+(B^3 y_0^3-s_k B y_0(y_2+y_3)+B^{-1}(y_0+s_j y_3))x_1\big)(s_k B y_0-s_j B^{-1})^{-1}\big)}\\
&=-3^k\eta(B^{-2}+B^2 y^2_0)\omega^{\Tr_k\big(B y_0(s_j y_0-s_k)+(s_k B^2 y_0^2+s_j y_0-y_2)^2(s_k B y_0-s_j B^{-1})^{-1}\big)}\\
&\ \times\omega^{-\Tr_k\big((B^3 y_0^3-B s_k y_0(y_2+y_3)+B^{-1}(y_0+s_j y_3))^2(B^{-2}+B^2 y^2_0)^{-1}\big)}\enspace,
\end{align*}
where $i$ is a complex primitive fourth root of unity and $\eta()$ is the quadratic character of $\fthreek$. 

Here we assumed that $s_k B y_0\neq s_j B^{-1}$ and $B^2 y^2_0\neq -B^{-2}$. Since $B^{-4}=\zeta^{2t(3^k+1)}\in\fthreek$ and $-1$ is a 
nonsquare in $\fthreek$ for odd $k$, then $y_0^2\neq -B^{-4}$, and in the case when $y_0=s_k s_j B^{-2}$ we have
\begin{align*}
&W_{\Tr_k(\tilde{g}_{y_0})}(y_3,y_2)=\\
&=\omega^{\Tr_k\big(s_j B^{-1}(B^{-2}-1)\big)}\sum_{x_1,x_2\in\fthreek}\omega^{\Tr_k\big(s_j B^{-1}x_1^2-(s_k B^{-2}+y_2-s_j B^{-1} x_1)x_2-(s_k B^{-2}+y_3)x_1\big)}\\
&=3^k\omega^{\Tr_k\big(s_j B^{-1}(s_k B^{-1}+y_2 B)^2-s_j(s_k B^{-2}+y_3)(s_k B^{-1}+y_2 B)\big)}\\
&=3^k\omega^{\Tr_k\big(s_j(y_2 B+s_k B^{-1})(y_2-y_3)\big)}\enspace.
\end{align*}
Thus, $\Tr_k(\tilde{g}_{y_0})$ is bent for any $y_0\in\fthreek$.

Compute the Walsh transform in point $-y$ using (\ref{eq:quad_Wtc}). 
\begin{align}
 \label{eq:b0notsp} 
\nonumber&W_f(-y)=\sum_{x\in\fthreen}\omega^{f(x)+\Tr_n(xy)}\\
\nonumber&=\sum_{x_0,x_1,x_2,x_3\in\fthreek}\omega^{\Tr_k\big(B^{-1}(x_1+x_2-x_3)x_0+y_0 x_0+g(x_1,x_2,x_3)+y_3 x_1+y_2 x_2+y_1 x_3\big)}\\
\nonumber&=3^k\sum_{x_1,x_2\in\fthreek}\omega^{\Tr_k\big(g(x_1,x_2,x_1+x_2+B y_0)+y_3 x_1+y_2 x_2+y_1(x_1+x_2+B y_0)\big)}\\
\nonumber&=3^k\omega^{\Tr_k\big(s_k B y_0(1-s_k y_0+s_k y_1)\big)}\\
\nonumber&\times\sum_{x_1,x_2\in\fthreek}\omega^{\Tr_k\big((B^{-1}+s_k B y_0)x_1^2+(B^{-1}-s_k B y_0)x_2^2-s_k B y_0 x_1 x_2+(y_3+y_1-y_0)x_1+(y_0+y_1+y_2+s_k B^2 y_0^2)x_2\big)}\\
\nonumber&=3^k\eta(B^{-1}+s_k B y_0) i^k 3^{k/2}\omega^{\Tr_k\big(B y_0(s_k-y_0+y_1)\big)}\\
\nonumber&\times\sum_{x_2\in\fthreek}\omega^{\Tr_k\big((B^{-1}-s_k B y_0)x_2^2+(y_0+y_1+y_2+s_k B y_0^2)x_2-(y_3+y_1-y_0-s_k B y_0 x_2)^2 (B^{-1}+s_k B y_0)^{-1}\big)}\\
\nonumber&=3^k\eta(B^{-1}+s_k B y_0) i^k 3^{k/2}\omega^{\Tr_k\big(B y_0(s_k-y_0+y_1)-(y_0-y_1-y_3)^2(B^{-1}+s_k B y_0)^{-1}\big)}\\
\nonumber&\times\sum_{x_2\in\fthreek}\omega^{\Tr_k\Big(\big((B^{-2}+ B^2 y_0^2)x_2^2+(B^3 y_0^3+s_k B y_0(y_2-y_3)+B^{-1}(y_0+y_1+y_2))x_2\big)(B^{-1}+s_k B y_0)^{-1}\Big)}\\
\nonumber&=-3^{2k}\eta(B^{-2}+B^2 y_0^2)\\
\nonumber&\times\omega^{-\Tr_k\big((y_0-y_1-y_3)^2 (B^{-1}+s_k B y_0)^{-1}+(B^3 y_0^3+s_k B y_0(y_2-y_3)+B^{-1}(y_0+y_1+y_2))^2 (B^{-2}+B^2 y_0^2)^{-1}\big)}\\
&\times\omega^{-\Tr_k\big((B^{-1}+s_k B y_0)^{-2}-B y_0(s_k-y_0+y_1)\big)}\enspace,
\end{align}
where $i$ is a complex primitive fourth root of unity and $\eta$ is the quadratic character of $\fthreek$. 

Here we assumed that $y_0\neq-s_k B^{-2}$ and $y_0^2\neq -B^{-4}$. The latter holds as noted above and the case when $y_0=-s_k B^{-2}$ we 
consider separately. If $y_0=0$ then 
\begin{align*}
W_f(-y)&=3^k\sum_{x_1,x_2\in\fthreek}\omega^{\Tr_k\big(B^{-1} x_1^2+B^{-1} x_2^2+(y_1+y_3)x_1+(y_1+y_2)x_2\big)}\\
&=-3^{2k}\omega^{-\Tr_k\big(B(y_1+y_2)^2+B(y_1+y_3)^2\big)}\enspace.
\end{align*}
If $y_0=-s_k B^{-2}$ then
\begin{align}
 \label{eq:b0sp}
\nonumber W_f(-y)&=3^k\omega^{\Tr_k(B^{-1}(1-s_k y_1))}\sum_{x_1,x_2\in\fthreek}\omega^{\Tr_k\big(-B^{-1} x_2^2+(B^{-1} x_2+y_1+y_3+s_k B^{-2})x_1+(y_1+y_2)x_2\big)}\\
&=3^{2k}\omega^{\Tr_k\big(B(y_1+y_3)(y_1-y_2-y_3)+s_k B^{-1}(y_3-y_2-y_1)\big)}\enspace.
\end{align}
Thus, function $f$ is bent and not weakly regular since the sign of the Walsh transform coefficients can be both plus and minus.

\begin{example}
Take the component function $\Tr_k(F(x))$ of (\ref{eq:trinom_nwr_vec}) with $k=1$, $j=2k$ and $t=(3^k-1)/2$ and find the $4$-variable 
polynomial representation for its dual. In this case, $B=1$ and combine (\ref{eq:b0notsp}) and (\ref{eq:b0sp}) with the indicator function 
of $x_0=-1$ to obtain the following polynomial of algebraic degree four 
\begin{align*}
&(x_0+1)^2(-(x_0-x_1-x_3)^2 (x_0+1)^{-1}-(x_0 x_2-x_0 x_3-x_0+x_1+x_2)^2 (x_0+1)^{-1}+x_0 x_1)\\
&\ -((x_0+1)^2-1)((x_1+x_3)(x_1-x_2-x_3)-x_1-x_2+x_3)\\
&=-(x_0-x_1-x_3)^2 (x_0+1)-(x_0 x_2-x_0 x_3-x_0+x_1+x_2)^2 (x_0+1)+x_0 x_1(x_0+1)^2\\
&\ -((x_0+1)^2-1)((x_1+x_3)(x_1-x_2-x_3)-x_1-x_2+x_3)\\
&=x_0^3+x_0^2(-x_1^2-x_1 x_2-x_1 x_3-x_2 x_3+x_1-x_2-x_3+1)\\
&\ +x_0(-x_1^2-x_2^2+x_1 x_2-x_1+x_3)+x_1^2-x_2^2-x_3^2+x_1 x_2+x_1 x_3\\
&=x_0^2(-x_1^2+x_1 x_2+x_1+x_2+1)-x_0(x_1^2+x_1-x_2+x_3-1)+x_1^2-x_2^2-x_1 x_2\enspace,
\end{align*}
where in the last identity we substitute $x_1=x_1-x_2$ and $x_2=x_2+x_3$. The corresponding relative trace form can be obtained using 
(\ref{eq:2univ}) and it consists of $12$ terms (dropping the linear term) 
\begin{align*}
\Tr_{4}\big(&-a^{13} x^{22}+a^{30} x^{20}-a^{31} x^{16}-a^{37} x^{14}+a^{14} x^{13}-a^{31} x^{11}\\
&+a^{30} x^{10}+a^{23} x^8-a^{22} x^7+a^{33} x^5-a^{35} x^4-a^{28} x^2\big)\enspace.
\end{align*}
This is an interesting example of dual bent function $f$ for which the dual $f^*$ has a different algebraic degree, which guarantees that 
$f$ and $f^*$ are EA-inequivalent. 
\end{example} 

\begin{example}
Take $k=3$ and $t=(3^k-1)/2$ and combine (\ref{eq:b0notsp}) and (\ref{eq:b0sp}) with the indicator function of $x_0=1$ to obtain the 
following polynomial of algebraic degree seven 
\begin{align*}
&\Tr_3\Big((x_0-1)^{26}\big(-(x_0-x_1-x_3)^2(x_0-1)^{-1}\\
&-(x_0^3-x_0 x_2+x_0 x_3+x_0+x_1+x_2)^2(x_0^2+1)^{-1}(x_0-1)^{-2}\\
&\ -x_0^2+x_0 x_1-x_0\big)-((x_0-1)^{26}-1)((x_1+x_2)(x_3-x_2-1)+x_3)\Big)\\
&=\Tr_3\big(-(x_0-x_1-x_3)^2(x_0-1)^{25}\\
&-(x_0^3-x_0 x_2+x_0 x_3+x_0+x_1+x_2)^2(x_0^2+1)^{25}(x_0-1)^{24}\\
&\ -x_0(x_0-x_1+1)(x_0-1)^{26}-((x_0-1)^{26}-1)((x_1+x_2)(x_3-x_2-1)+x_3)\big)\\
&=\Tr_3\big(-(x_0-x_1-x_3+1)^2x_0^{25}-(x_0^3-x_0 x_2+x_0+x_1+x_3-1)^2(x_0^2-x_0-1)^{25}x_0^{24}\\
&\ -(x_0+1)(x_0-x_1-1)x_0^{26}+(x_0^{26}-1)((x_1+x_2+x_3)x_2+x_1+x_2)\big)\enspace,
\end{align*}
where in the last identity, we substitute $x_0=x_0+1$ and $x_2=x_2+x_3$. 
\end{example} 

\section*{Acknowledgment}
The authors would like to thank Ay\c{c}a \c{C}e\c{s}melio\u{g}lu and Wilfried Meidl for making us think about vectorial bent functions in 
Theorem~\ref{th:trinom_nwr_vec}, also Tor Helleseth and Chunlei Li for useful discussions. 

\bibliographystyle{splncs04}
\bibliography{IEEEabrv,mrabbrev,all}

\end{document}